\documentclass[12pt]{article}
\usepackage{amsmath}
\usepackage{graphicx}
\usepackage{enumerate}
\usepackage{url} % not crucial - just used below for the URL 

\newcommand{\blind}{1}

\usepackage{amssymb, amsthm, color, mathrsfs,eufrak, tikz}

\usepackage{pdfpages, rotating}
\usepackage[export]{adjustbox}% http://ctan.org/pkg/adjustbox

\input{tikzlibrarybayesnet.code}
\usepackage{bm, bbm}
\usepackage{csvsimple}
\usepackage[width=\linewidth, font = small]{caption}
\usepackage{subcaption}
\usepackage{hyperref}
\usepackage[linesnumbered,ruled]{algorithm2e}
\SetKwInput{KwData}{Input}

\newcommand\independent{\protect\mathpalette{\protect\independenT}{\perp}}
\def\independenT#1#2{\mathrel{\rlap{$#1#2$}\mkern2mu{#1#2}}}

\definecolor{firebrick1}{HTML}{FF3030}
\definecolor{dodgerblue}{HTML}{1E90FF}

\usepackage{listings}

\newcommand{\br}{{\mathbb R}}

\theoremstyle{definition}
\newtheorem{proposition}{Proposition}%[section]
\newtheorem*{definition}{Definition}
\newtheorem{corollary}[proposition]{Corollary}
\newtheorem{theorem}{Theorem}
\newtheorem{lemma}{Lemma}

\newcommand{\eps}{{\epsilon}}

\newcommand{\norm}[1]{\left \lVert #1 \right \rVert}

\newcommand{\fdp}{\textnormal{FDP}}
\newcommand{\fdr}{\textnormal{FDR}}

\newcommand{\fdphat}{\widehat{\textnormal{FDP}}}
\newcommand{\cR}{\mathcal R}
\newcommand{\cU}{\mathcal U}
\newcommand{\sU}{\mathscr U}

\newcommand{\cH}{\mathcal H}
\newcommand{\cT}{\mathcal T}
\newcommand{\F}{\mathfrak F}
\newcommand{\EE}[1]{\mathbb E\left[#1\right]}
\newcommand{\PP}[1]{\mathbb P\left[#1\right]}

\newcommand{\PPst}[2]{\mathbb P\left[\left. #1 \right| #2 \right]}
\newcommand{\cM}{\mathcal M}
\newcommand{\ind}{\mathbbm 1}

\begin{document}

\def\spacingset#1{\renewcommand{\baselinestretch}%
{#1}\small\normalsize} \spacingset{1}

%%%%%%%%%%%%%%%%%%%%%%%%%%%%%%%%%%%%%%%%%%%%%%%%%%%%%%%%%%%%%%%%%%%%%%%%%%%%%%

\if1\blind
{
%  \title{\bf FDR control following filtering for hierarchically structured hypotheses}
\title{\bf Filtering the rejection set while preserving false discovery rate control}
%	Eugene Katsevich, Chiara Sabatti, Marina Bogomolov}
\author{Eugene Katsevich\thanks{
	EK acknowledges support by the Hertz Foundation; CS by NSF DMS 1712800, NIH R01MH113078, and the Stanford Discovery Innovation Funds; and MB by the Israel Science Foundation grant no. 1112/14. The phenome-wide association study data analysis was conducted using the UK Biobank Resource (application number 27837).}\hspace{.2cm}\\
Department of Statistics and Data Science,\\
Carnegie Mellon University\\
and\\
Chiara Sabatti\\
Departments of Statistics and Biomedical Data Science,\\ 
Stanford University\\
and\\
Marina Bogomolov\\
Faculty of Industrial Engineering and Management,\\ 
Technion - Israel Institute of Technology}
\maketitle
} \fi

\if0\blind
{
\bigskip
\bigskip
\bigskip
\begin{center}
{\LARGE\bf Title}
\end{center}
\medskip
} \fi

\bigskip
\begin{abstract}
Scientific hypotheses in a variety of applications have  domain-specific structures, % have hierarchical, spatial, or group structures, 
such as the tree structure of the International Classification of Diseases (ICD), the directed acyclic graph structure of the Gene Ontology (GO), or the spatial structure in genome-wide association studies. In the context of multiple testing, the resulting relationships among hypotheses can create redundancies among rejections that hinder interpretability. This leads to the practice of filtering rejection sets obtained from multiple testing procedures, which may in turn invalidate their inferential guarantees. We propose Focused BH, a simple, flexible, and principled methodology to adjust for the application of any pre-specified filter. We prove that Focused BH controls the false discovery rate under various conditions, including when the filter satisfies an intuitive monotonicity property and the p-values are positively dependent. We demonstrate in simulations that Focused BH performs well across a variety of settings, and illustrate this method's practical utility via analyses of real datasets based on ICD and GO.
\end{abstract}

\noindent%

%\newpage
\spacingset{1.5} % DON'T change the spacing!

%%%%%%%%%%%%%%%%%%%%%%%%%%%%%%%%%%%%%%%%%%%%%%%%%%%

\section{Testing hierarchically structured hypotheses} \label{sec:intro}

Hierarchical structures are ubiquitous in biomedical applications. Consider for example the International Classification of Diseases (ICD), where roughly 20,000 disease codes are structured as a tree (in fact, a forest of trees) with levels of increasing specificity. For instance, the more specific term ``ankylosing spondylitis" is a child of the more general term ``spondylopathies" (vertebrae disorders). Another example is the Gene Ontology (GO) \cite{AetE00}, a hierarchically organized structure containing thousands of biological processes (GO terms), each annotated with several genes known to be involved in that process. In many cases, each node in the hierarchical structure corresponds to a scientific hypothesis. In phenome-wide association studies (PheWAS) \cite{DetC10}, researchers explore the possible connection between a genetic marker and each ICD code. In Gene Ontology enrichment analysis, the goal is to interpret sets of differentially expressed genes by finding biological processes whose annotated genes significantly overlap with these gene sets. Therefore, we obtain multiple testing problems where hypotheses are identified with the nodes in a hierarchical structure.
%Other examples of such problems can be found in Section \ref{sec:discussion}.

In such structured multiple testing problems, the hypotheses are %overlapping and 
interrelated. For example, in some cases hypotheses are linked via logical implications, such that if a node is null then all its descendants are also null, see~(\ref{self-contained}). % as as happens in the case where each parent hypothesis is the intersection of all its child hypotheses. 
Whether due to these or other relationships among nodes, there is some redundancy built into the testing problem. This redundancy is difficult to remove before seeing the data, because the best choice of hypotheses to keep is likely to depend on the data. For example, when the  logical relationships exist between nodes and their descendants, the nodes further down may be more interesting to reject as they carry information at ``higher resolution." But there is often less power to make higher-resolution rejections, so the power of choosing a resolution in advance  depends heavily on one's ability to guess the location of the signal before having seen the data. 

Therefore, it is common in hierarchical settings to apply standard multiple testing procedures, like the Benjamini-Hochberg (BH) algorithm \cite{BH95}, and then take care of the redundancy post hoc. Once a rejection set is in hand, it may be much clearer which nodes should be kept and which should be discarded. We refer to the process of choosing a subset of rejections to keep as \textit{filtering}. For example, the \textit{outer nodes filter} \cite{Y08, GS11, MG15_Holm} restricts attention to those rejected hypotheses for which no descendants have also been rejected (see Figure \ref{fig:toy_example}). This filter is most natural for trees. For more complex structures like GO, domain-specific filters %like REVIGO \cite{SetS11} 
have been developed. Given any (idempotent) filter, we can define a \textit{non-redundant} rejection set as one that does not change when filtered.

Applying a filter to the output of a multiple testing procedure can raise issues, however, because in some cases the filtering step invalidates the inferential guarantee of the testing procedure. In particular, this problem arises when the error rate being targeted is the \textit{false discovery rate} (FDR): the expected proportion of false discoveries in the rejection set. For example, outer nodes filtering tends to inflate the proportion of false discoveries (FDP) because it has a preference for nodes lower in the tree; these are more informative but, given logical relationships, are more often null. Figure \ref{fig:toy_example} illustrates this issue with a toy example. 	Therefore, applying an FDR procedure like BH and then focusing on outer node rejections is an intuitive but misleading attempt to control the FDR of the latter set. 

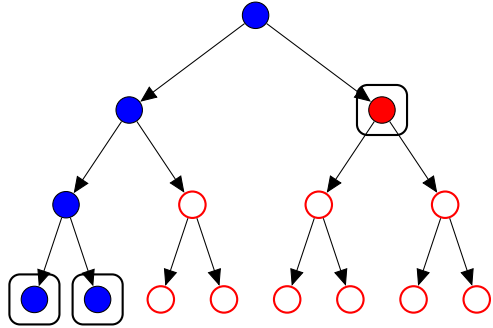
\begin{figure}[h!]
\centering

\begin{tikzpicture}[scale = 0.21]
\node[obs, text opacity=1,fill=blue,text opacity=1, minimum size = 10pt]  (H11)  at(1,1)   {} ; 
\node[obs,  text opacity=1,fill=blue,text opacity=1, minimum size = 10pt]  (H12)  at(5,1)  {} ; 
\node[obs, text opacity=1,fill=white,text opacity=1,thick, draw = red, minimum size = 10pt]  (H13)  at(9,1)   {} ; 
\node[obs,  text opacity=1,fill=white,text opacity=1, thick, draw = red, minimum size = 10pt]  (H14)  at(13,1)  {} ; 
\node[obs, text opacity=1,fill=white,text opacity=1, thick, draw = red, minimum size = 10pt]  (H15)  at(17,1)   {} ; 
\node[obs,  text opacity=1,fill=white,text opacity=1, thick, draw = red, minimum size = 10pt]  (H16)  at(21,1)  {} ; 
\node[obs, text opacity=1,fill=white,text opacity=1, thick, draw = red, minimum size = 10pt]  (H17)  at(25,1)   {} ; 
\node[obs,  text opacity=1,fill=white,text opacity=1, thick, draw = red, minimum size = 10pt]  (H18)  at(29,1)  {} ; 

\node[obs, text opacity=1,fill=blue,text opacity=1, minimum size = 10pt]  (H21)  at(3,7)   {} ; 
\node[obs,  text opacity=1,fill=white,text opacity=1, thick, draw = red, minimum size = 10pt]  (H22)  at(11,7)  {} ; 
\node[obs, text opacity=1,fill=white,text opacity=1, thick, draw = red, minimum size = 10pt]  (H23)  at(19,7)   {} ; 
\node[obs,  text opacity=1,fill=white,text opacity=1, thick, draw = red, minimum size = 10pt]  (H24)  at(27,7)  {} ; 

\node[obs, text opacity=1,fill=blue,text opacity=1, minimum size = 10pt]  (H31)  at(7,13)   {} ; 
\node[obs,  text opacity=1,fill=red,text opacity=1, minimum size = 10pt]  (H32)  at(23,13)  {} ; 

\node[obs, text opacity=1,fill=blue,text opacity=1, minimum size = 10pt]  (H41)  at(15,19)   {} ; 

\edge {H41} {H31}
\edge {H41} {H32}

\edge {H31} {H21}
\edge {H31} {H22}
\edge {H32} {H23}
\edge {H32} {H24}

\edge {H21} {H11}
\edge {H21} {H12}
\edge {H22} {H13}
\edge {H22} {H14}
\edge {H23} {H15}
\edge {H23} {H16}
\edge {H24} {H17}
\edge {H24} {H18}

\node[fit=(H11), draw, rounded corners, thick, inner sep=1.5mm, black] {};
\node[fit=(H12), draw, rounded corners, thick, inner sep=1.5mm, black] {};
\node[fit=(H32), draw, rounded corners, thick, inner sep=1.5mm, black] {};

\end{tikzpicture}

\caption{Outer nodes filtering doubles the false discovery proportion from 1/6 to 1/3. Blue nodes correspond to  non-null hypotheses and red nodes to null ones. A node is filled in when its hypothesis is rejected. Outer nodes are circled.}
\label{fig:toy_example}
\end{figure}

Few methods have been developed to control the FDR of the outer node rejections. Most multiple testing methods for hierarchically structured hypotheses focus instead on either ensuring notions of consistency with respect to the hierarchy or boosting power by leveraging logical relationships \cite{LG16, FBR15, DAGGER, BetS17a, CetM17,katsevich17mkf}. The only such method we are aware of is one proposed by Yekutieli \cite{Y08}, which controls the outer node FDR for tree-structured hypotheses. However, a limitation of this method is that it requires p-values to be independent across the tree. Another class of methods \cite{M08, MG15, MG15_Holm} controls the \textit{family-wise error rate} (FWER), or probability of making any false discoveries, for hierarchically structured hypotheses. Unlike the FDR, controlling the FWER before filtering also leads to FWER control (and thus FDR control) after filtering. However, controlling the FWER can be too stringent a target for modern large-scale testing applications. 

We propose Focused BH, a novel FDR-controlling method %for hierarchically structured hypotheses 
that inputs any pre-specified filter and outputs a non-redundant rejection set with respect to that filter. We provide a broad set of conditions on the filter and p-value distribution under which Focused BH provably controls the FDR. As an important corollary of our result, we can show that Focused BH applied to arbitrary tree structures with the outer nodes filter controls the FDR if positively dependent (PRDS) p-values are available. This result is practically significant because to our knowledge it is the first procedure targeting outer nodes FDR control while allowing p-value dependency. Though motivated by hierarchical applications, Focused BH is actually a general-purpose methodology for controlling FDR in situations where a pre-specified filter is used to focus attention on a subset of discoveries. Thus, the scope of our work extends beyond hierarchically structured hypotheses. 
For example, filtering is also relevant to applications with spatial structure, such as genome-wide association studies (GWAS) or imaging applications. We formally set up the problem in Section~\ref{sec:problem_setup} and then present Focused BH and its theoretical properties in Section \ref{sec:method}. We then test Focused BH on simulated and real data (Sections \ref{sec:experiments} and \ref{sec:data_analysis}, resp.), and end with a discussion in Section \ref{sec:discussion}. Code to reproduce the simulations and data analysis from Sections~\ref{sec:experiments} and \ref{sec:data_analysis} can be found at \url{https://github.com/ekatsevi/Focused-BH}. 

\section{Problem setup} \label{sec:problem_setup}

\subsection{Hierarchically structured hypotheses} \label{sec:setup_hierarchical}

We encode hierarchical structures as \textit{directed acyclic graphs} (DAGs). A directed graph $\mathcal G$ is a collection of \textit{nodes} labeled $[m] = \{1, \dots, m\}$ and directed edges $j_1 \rightarrow j_2$ for some pairs $(j_1, j_2) \in [m] \times [m]$, and $\mathcal G$ is acyclic if there is no cycle among the edges in $\mathcal G$. A DAG has a set of root nodes (nodes with no incoming edges) at the first level, and the remaining nodes can be arranged in levels based on their depth, defined as the maximum distance from a root node, plus one. For the purposes of this paper, we define a \textit{tree} as a DAG where each node has at most one parent. This definition formally includes both trees with a single root node and forests of trees, but we ignore this distinction because it does not play a role in this paper. The ICD is a tree, and the GO is a non-tree DAG. 

Next, we give one example of how hypotheses at each node of a DAG may arise. Suppose we have a set of $K$ \textit{items}, indexed by $[K]$, such that each node $j$ in $\mathcal G$ is associated with a set of items $\mathcal A_j \subseteq [K]$. Items are basic units of observation or inference, which comprise higher-level concepts represented by nodes. In the GO, items are genes, and $\mathcal A_j$ is the set of genes annotated to the biological process represented by GO term $j$. The edges in the DAG structure correspond to inclusion relationships among sets of items:
\begin{equation}
\text{if } j_1 \leadsto j_2 \text{ for nodes } j_1,j_2,\ \text{then } \mathcal A_{j_1} \supseteq \mathcal A_{j_2},
\label{inclusion_relationships}
\end{equation}
where $j_1 \leadsto j_2$ indicates that $j_2$ is a descendant of $j_1$ in $\mathcal G$. If we have item-level hypotheses $H_1^{\text{item}}, \dots, H_K^{\text{item}}$, then,  the \textit{self-contained} null for each node $j$ is the intersection null
\begin{equation}
H^{\text{node}}_{j} \equiv \bigcap_{k \in \mathcal A_j} H^{\text{item}}_k.
\label{self-contained}
\end{equation}
The \textit{self-contained} null was defined by \cite{GB07} in the context of GO testing, 
%testing, 
and was contrasted with a \textit{competitive} null, which we address in Section~\ref{sec:data_analysis}.  %In the context of gene set analysis in GO, the null hypothesis above states that no genes in $\mathcal A_j$ are differentially expressed, and this is precisely the \textit{self-contained} hypothesis according to \cite{GB07}, who address the methodological issues of GO testing, and contrast this null with the \textit{competitive null} hypothesis, which we do not discuss here.  %definition was given by \cite{GB07} in the context of gene set analysis in GO, 
We abbreviate $H_j \equiv H_j^{\text{node}}$. 
With this definition, the inclusion relationships~(\ref{inclusion_relationships}) among sets of items lead to logical relationships among the corresponding node-level hypotheses:
\begin{equation}
\text{if } {j_1} \leadsto {j_2}, \text{ then } H_{j_1} \Rightarrow H_{j_2},
\label{logical_relationships}
\end{equation}
%In other words, 
i.e. if the %self-contained 
null hypothesis for a node is true, then it is true for all  its descendants. %
%it must also be true for any of its descendants. 
Such logical relationships are natural, but not universal, among hierarchical %multiple
testing problems. 

\subsection{Multiple testing and filtering}

Given a choice of hypothesis to test for each node, we arrive at a set $\mathcal H = \{H_1, \dots, H_m\}$ of hypotheses identified with the nodes of a DAG $\mathcal G$. Suppose a vector of p-values $\bm p = (p_1, \dots, p_m)$ is available to test these hypotheses. For example, in the setup described in the previous section, we may have item-level p-values $p_k^{\text{item}}$. These can then be aggregated using global tests to get node-level p-values $\bm p$ for the self-contained hypotheses~\eqref{self-contained}.

A multiple testing procedure, which we denote by $\cM_0$, inputs $\bm p$ and outputs a rejection set $\cR^* \subseteq \mathcal H$. The redundancy in this rejection set would be remedied by a filter, such as
\begin{equation}
\cU^* \equiv \F_{\mathcal G}(\cR^*) \equiv \{j \in \cR^*: \text{no descendants of $j$ in $\mathcal G$ are also in } \cR^*\},
\end{equation}
where $\F_{\mathcal G}$ is the outer nodes filter introduced already. REVIGO \cite{SetS11}  is a more involved filter tailored for GO that performs an agglomerative clustering algorithm on a set of GO terms based on semantic similarity between terms. The representatives for each cluster are chosen based on the graph structure as well as the p-values at each node. As different applications call for different filters, we work with the following very general definition. %of a filter.

\begin{definition}
Given a vector of p-values $\bm p$ and a subset $\cR \subseteq \mathcal H$, a \textbf{filter} $\F$ is any map $${\F}: (\cR, \bm p) \mapsto \cU \subseteq \mathcal H$$ with the property that $\cU \subseteq \cR$. 
\end{definition}

\noindent Note that $\F$ has as arguments both the set $\cR$ to be filtered and the p-values $\bm p$ (we use asterisks, as in $\cR^*$ or $\mathcal U^*$, to denote outcomes of multiple testing procedures; for dummy variables we use just $\cR$ and $\mathcal U$). In practice, the filter would be applied to the rejection set $\cR^*$ of a multiple testing procedure, which itself is a function of the p-values. However, writing explicitly the dependence of the filter $\F$ on $\cR$ allows us to separate the definition of a rejection set $\cR^*$ from that of  a filter, and it is useful to emphasize whether the p-values contribute to the output $\cU^*$ only via ${\cR^*}$ or otherwise. For example, \textit{fixed filters} are filters $\boldmath \F$ such that $\F(\cR, \bm p) = \F_0(\cR)$ for some function $\F_0$, i.e. $\cU^* = \F(\cR^*, \bm p)$ does not depend on the data once $\cR^*$ is identified. The outer nodes filter is a special case of a fixed filter, with $\F_0 = \F_{\mathcal G}$. REVIGO is an example of a filter that uses the information in $\bm p$. Another example is a \textit{screening filter}: a filter $\F$ such that $\cU=\cR \cap \mathcal S(\bm p)$, where $\mathcal S: \bm p \rightarrow \mathcal S_0 \subseteq \mathcal H$ is called screening function. See \cite{BB14, BetS17a} for applications of such filters in fMRI and GWAS. 

Given a filter $\F$, we consider multiple testing procedures $\mathcal M: \bm p \rightarrow \cU^*$ of the form 
\begin{equation}
\mathcal M = \F \circ \mathcal M_0,
\label{composition}
\end{equation}
where $\mathcal M_0: \bm p \rightarrow \cR^*$ is a ``base procedure" whose output is then filtered via $\F: \cR^* \rightarrow \cU^*$. Our goal is to control the FDR of $\mathcal M$, i.e.
\begin{equation*}
\fdr(\cM) \equiv \EE{\fdp(\cU^*)} \equiv \EE{\frac{|\cU^* \cap \cH_0|}{|\cU^*|}}\leq q,
\end{equation*}
where $\mathcal H_0 \subseteq \mathcal H$ denotes the set of true nulls and by convention $0/0 \equiv 0$. As discussed in the introduction, the FDR control of the composite procedure $\mathcal M$ is usually not guaranteed if the base procedure $\cM_0$ controls the FDR (e.g., see Figure~\ref{fig:BH_experiments}). Instead, special care must be taken to account for the extra filtering step. This motivates us to propose Focused BH, a procedure $\cM$ of the form \eqref{composition} that provably controls the FDR for a broad class of filters.

\section{Focused BH} \label{sec:method}

Our goal is to allow the incorporation of arbitrary filters into the multiple testing process. This includes complex filters like REVIGO, which may only be available to us as black box software packages. We would like to handle hierarchical testing problems where logical relationships~\eqref{logical_relationships} might or might not hold, as well as non-hierarchical problems where these relationships are not even well-defined. Therefore, we design Focused BH without building in special properties of certain filters, and without relying on hierarchical or logical relationships among hypotheses. We present this methodology next.

\subsection{Methodology} \label{sec:FocusedBH-method}

Consider a collection of base rejection sets $\cR$ corresponding to p-value thresholds $t \in [0,1]$:
\begin{equation}
\cR(t, \bm p) \equiv \{j: p_j \leq t\}. \label{BHreject}
\end{equation}
When composed with the filter $\F$, we obtain a collection of filtered rejection sets 
\begin{equation*}
\mathcal U(t, \bm p) \equiv \F(\cR(t, \bm p), \bm p) \subseteq \mathcal H.
\end{equation*}
We construct Focused BH by finding a data-dependent threshold $t^* \in [0,1]$ and defining $\cU^* \equiv \cU(t^*, \bm p)$. We must choose this threshold so that FDR control is guaranteed. To ensure proper calibration, we employ a conservative estimate of $V(t) \equiv |\cU(t, \bm p) \cap \mathcal H_0|$. To this end, note that for a given $t$, we have $\cU(t, \bm p) \subseteq \cR(t, \bm p)$, so $|\cU(t, \bm p) \cap \mathcal H_0| \leq |\cR(t, \bm p) \cap \mathcal H_0|$. This suggests the following definition of $\widehat V$:
\begin{equation}
\mathbb E\left[|\cU(t, \bm p) \cap \mathcal H_0|\right] \leq \mathbb E\left[|\cR(t, \bm p) \cap \mathcal H_0|\right] = m_0\cdot t \leq m \cdot t \equiv \widehat V(t),
\label{BH_derivation}
\end{equation}
where $m_0 \equiv |\mathcal H_0|$. Note that the equality holds under the assumption of uniform null p-values, and becomes an inequality for null p-values that are stochastically larger than uniform. This leads to the following estimate for the FDP of the set $\cU(t, \bm p)$:
\begin{equation}
\fdphat(t) \equiv \frac{m \cdot t}{| \cU(t, \bm p)|}.
\label{FDP_hat}
\end{equation}
We choose the maximum threshold such that the FDP estimate is below the target level $q$:
\begin{equation}
t^* \equiv \max\{t \in \{0, p_1, \dots, p_m\}: \fdphat(t) \leq q\}.
\label{t_star}
\end{equation}
Note that the set in (\ref{t_star}) is always nonempty because $\fdphat(0) = 0$. This leads to Procedure~\ref{focused_BH}, which applies for any filter $\F$.
\noindent
\begin{center}
\begin{minipage}{0.7\linewidth}
\begin{algorithm}[H]
	\SetAlgorithmName{Procedure}{}\; %last arg is the title of listing table		
	\KwData{p-values $p_1, \dots, p_m$, filter $\F$}
	\For{$t \in \{0, p_1, \dots, p_m\}$}{
		Compute $\displaystyle \fdphat(t) = \frac{m \cdot t}{|\F(\{j: p_j \leq t\}, \bm p)|}$\;
	}	
	Compute $t^* \equiv \max\{t \in \{0, p_1, \dots, p_m\}: \fdphat(t) \leq q\}$\;
	Find the base rejection set $\cR^* = \{j: p_j \leq t^*\}$\;
	Compute $\cU^* = \F(\cR^*, \bm p)$\; 		
	\KwResult{Filtered rejection set $\cU^*$.}
	\caption{\bf Focused BH}
	\label{focused_BH}
\end{algorithm}
\end{minipage}
\end{center}

This procedure is similar to the empirical Bayes formulation of BH proposed by Storey  et al. \cite{Storey04}, and in fact reduces to it when $\F$ is trivial, i.e. $\F(\cR, \bm p) = \cR$. The name Focused BH  reflects the fact that Procedure~\ref{focused_BH} is a generalization of BH and provides guarantees on the set of discoveries scientists decide to focus upon. Next, we move on to stating sufficient conditions on the filter and the p-value dependence structure for Focused BH to control the FDR. 

\subsection{FDR control results} \label{sec:theoretical_results}
Recall that the $p$-values $\bm p$ are valid if $\PP{p_j \leq t} \leq t \ \text{for all } t \in [0,1]$ and all $j \in \mathcal H_0$. Note that this definition includes uniform p-values. %or the item-level p-values which are used to construct  $\bm p$ are   \textit{valid}, i.e. their coordinates corresponding to true null hypothesis are uniform or are stochastically larger than uniform. %i.e. $\PP{p_j \leq t} \leq t \ \text{for all } t \in [0,1]$ and all $j \in \mathcal H_0$. Note that this definition includes uniform p-values. If the $p$-values are obtained from item-level p-values $\bm p^{\text{item}},$ then it is assumed that $\bm p^{\text{item}}$ are valid.
Let $\bm p_{-j}$ denote all p-values but the $j$th one.
The following is our main theoretical result, establishing FDR control for Focused BH. 
\begin{theorem} \label{main_theorem}
Focused BH controls the FDR at level $q$ under any of the following conditions, assuming that the p-values $\bm p$ are valid in (i) and (iii):
\begin{itemize}
\item[(i)] The p-values  $\bm p$ are PRDS \cite{BY01} and $\F$ is \textit{monotonic}, i.e. for any $\bm p^1 \leq \bm p^2$ (component-wise) and $\cR^1 \supseteq \cR^2$ we have $|\F(\cR^1, \bm p^1)| \geq |\F(\cR^2, \bm p^2)|$;
\item[(ii)] The p-values are obtained from valid item-level p-values $\bm p^{\text{item}}$ via the Simes test \cite{S86} of the self-contained null~\eqref{self-contained} (recall Section~\ref{sec:setup_hierarchical}), 
$\bm p^{\text{item}}$ are PRDS, and $\F$ is monotonic.
\item[(iii)] %$p_j \independent \bm p_{-j}$ 
For all $j \in \mathcal H_0$, $p_j$ is independent of $\bm p_{-j},$ 
and the filter is either monotonic or of the form $\F(\cR, \bm p) = \F_0(\cR \cap \mathcal S(\bm p))$, where $\F_0$ is an arbitrary fixed filter and $\mathcal S$ is a \textit{stable} screening function \cite{BH13}, i.e. the set of screened hypotheses $\mathcal S(\bm p)$ does not change if for any $j$ we fix $\bm p_{-j}$ and vary $p_j$ as long as $j \in \mathcal S(\bm p)$.% (Note: $\bm p_{-j}$ denotes all p-values but the $j$th one.)
%		\textit{simple}, i.e. the quantity $|\F(\cR, \bm p)|$ does not depend on $\bm p$.
\end{itemize}
\end{theorem}
The proof of this result relies on previously established techniques \cite{BY01, BR08, BB14, RetJ17}; see supplementary Section \ref{sec:supp_proofs} for all proofs. Part (i) of the theorem allows for a notion of positive p-value dependence called PRDS, making it well-suited for application to structured settings. Part (ii) is a similar statement applying to self-contained hypotheses, requiring that item-level p-values be PRDS. The monotonicity requirement on the filter in parts (i) and (ii) is not too restrictive; one expects many reasonable filters to satisfy this property. For example, the outer nodes filter is monotonic on trees, leading to the following corollary.

\begin{corollary} \label{outer_nodes_corollary}
Given an arbitrary tree $\mathcal T$, Focused BH with the outer nodes filter $\F(\cR, \bm p) = \F_{\mathcal T}(\cR)$ controls the FDR at level $q$ if $\bm p$ is PRDS. 
\end{corollary}

Part (iii) of Theorem~\ref{main_theorem} accommodates a broader class of filters at the cost of a more restrictive assumption on the p-value dependence structure (independence is a special case of PRDS). For example, the outer nodes filter is not monotonic on general DAGs, because adding a rejection at a node with two rejected outer node parents will decrease the total number of outer nodes. Nevertheless, since it is a fixed filter, it satisfies the assumption of part (iii) of Theorem~\ref{main_theorem} for any DAG structure, so we get FDR control under the more stringent assumption of independence. Other filters may be more appropriate for DAG structures; we propose a monotonic filter on general DAGs in supplementary Section~\ref{sec:supp_soft_outer_nodes}. For arbitrary filters and dependencies among the $p$-values, one may revert to a more conservative variant of Focused BH, given in supplementary Section \ref{sec:reshaping}.  We remark that \cite{BR08}, whose results we use to prove part (i) of Theorem \ref{main_theorem}, foresaw procedures like Focused BH. In Section 4.5 of their paper they discuss scenarios where traditional step-up procedures cannot be applied because ``additional constraints come into play." 

%While we design Focused BH to control the post-filtering FDR, %of the filtered rejection set, 
%it is often the case that the corresponding base procedure (i.e. Focused BH without the final filtering step) 
%also controls the FDR. This can be important %in applications where 
%when the rejection set is inspected both before and after filtering. 
While Focused BH is designed to control the FDR post-filtering, it is often the case that its base procedure, which rejects the hypotheses in $\mathcal{R}^*,$ also controls the FDR. 
\begin{theorem} \label{secondary_theorem}
Suppose any of the following conditions holds, assuming $\bm p$ are valid in (iii):
\begin{itemize}
\item[(i)] The assumptions of Theorem~\ref{main_theorem} part (i) hold;
\item[(ii)] The assumptions of Theorem~\ref{main_theorem} part (ii) hold;
\item[(iii)] The p-values $\bm p$ are independent, and the quantity $|\F(\cR, \bm p)|$ does not depend on $\bm p$.
\end{itemize}
Then, the Focused BH base procedure controls the FDR at level $q$: $\EE{\fdp(\cR(t^*, \bm p))} \leq q$. 
\end{theorem}
This result can be important in applications where 
the rejection set is inspected both before and after filtering. %Comparing the conditions of 
Theorems above show %~\ref{main_theorem} and \ref{secondary_theorem} show
the FDR is controlled for these both sets
for PRDS $p$-values and monotonic filters, or independent $p$-values and fixed filters.

%This result should not be surprising, since as we have argued, the FDR of $\cR^*$ is generally easier to control than the FDR of $\cU^*$, so controlling the latter will often also result in control of the former. Comparing the conditions of Theorems~\ref{main_theorem} and \ref{secondary_theorem}, we see this happens for PRDS p-values and monotonic filters, or independent p-values and fixed filters.

\subsection[Improving power via resampling]{Improving the power of Focused BH} \label{sec:improving}

Note that the estimate $\widehat V$ derived in (\ref{BH_derivation}) does not account for the filter and is in fact the same as the BH estimate. This can make Focused BH conservative if the filter substantially reduces the number of rejected nulls. No matter what the signal configuration, note that $t^*_{\text{FBH}} \leq t^*_{\text{BH}} \text{ almost surely}$, where $t^*_{\text{FBH}}$ and $t^*_{\text{BH}}$ are the p-value thresholds of Focused BH and BH, respectively. This is the case because $|\F(\cR(t, \bm p))| \leq |\cR(t, \bm p)|$ for each $t$. In other words, Focused BH is more conservative than BH. This conservativeness is necessary only to the extent that the filter preferentially chooses nulls over non-nulls. 

To improve the power of Focused BH, we must design an estimate of $V(t)$ with less upward bias. As a theoretical benchmark, we may define an oracle estimate of $V(t)$ via
\begin{equation}
\mathbb E[V(t)] = \mathbb E\left[|\cU(t, \bm p) \cap \mathcal H_0|\right] \equiv \widehat V^{\text{oracle}}(t).
\label{V_hat_oracle}
\end{equation}
By construction, $\widehat V^{\text{oracle}}(t)$ is an unbiased estimate for $V(t)$ (i.e., $\mathbb E[\widehat V^{\text{oracle}}(t)] = \EE{V(t)}$; we use the term ``unbiased" loosely since $V(t)$ is itself random), and thus it accounts for the filter's possible reduction of the number of rejected nulls. Of course, this estimator cannot be computed in practice because it requires access to the ground truth data-generating distribution. Therefore, we propose a resampling approach to approximate it. 

Suppose we have a mechanism to generate p-value vectors $\tilde{\bm p}$ from the global null distribution. This may be the case if we know the null distribution. Otherwise, a permutation approach may be applicable. To illustrate, suppose the p-values were obtained from a data set $\mathcal D$, and that there is a group of permutations acting on the data. For example, suppose that $\bm X \in \br^{n \times K}$ is a gene expression matrix, and $\bm Y \in \{0,1\}^n$ is a binary indicator vector of treatment versus control. Each gene receives a p-value based on, say, a two-sample t-test comparing the expression of this gene in cases and controls. In this setting, the data is $\mathcal D = (\bm X, \bm Y)$ and can be permuted by permuting the labels $\bm Y$: $\tilde{\mathcal D}^{b} = (\bm X, \bm{\tilde Y^{b}})$. Now, suppose that we have $B$ permutations of the data $\tilde{\mathcal D}^{1}, \dots, \tilde{\mathcal D}^{B}$, resulting in p-value vectors $\tilde{\bm p}^1, \dots \tilde{\bm p}^B$ and corresponding filtered sets $\cU(t, \tilde{\bm p}^1), \dots, \cU(t, \tilde{\bm p}^B)$. Then, in certain cases, we may write
\begin{equation}
\begin{split}
\mathbb E[V(t)] = \mathbb E\left[|\cU(t, \bm p) \cap \mathcal H_0|\right] \overset?= \mathbb E\left[|\cU(t, \tilde{\bm p}) \cap \mathcal H_0|\right] \leq \EE{|\cU(t, \tilde{\bm p})|} \approx \frac{1}{B}\sum_{b = 1}^B |\mathcal U(t, \tilde{\bm p}^b)| \equiv \widehat{V}^{\text{perm}}(t).
\label{permutation_derivation}
\end{split}
\end{equation}
The validity of the equality with the question mark depends on the filter and the permutation mechanism; we discuss sufficient conditions for this equality to hold in Section~\ref{sec:supp_perm} of the supplement. For those situations when the above derivation is valid, $\widehat V^{\text{perm}}(t)$ can be a much better estimate than $\widehat V(t) = m\cdot t$ because it incorporates the filter into its definition. We do not yet have FDR control results for this permutation-based version of Focused BH; in particular note that \eqref{permutation_derivation} is not sufficient for FDR control. Nevertheless, we have found that it performs well in our 
%numerical 
simulations; see Section~\ref{sec:experiments}. We remark that there has not been much work on permutation-based FDR control in general, though %\cite{YB99} and 
\cite{YB99,TC01,HG18} are relevant. 

Another way to improve the power of Focused BH is to account for the non-null proportion. In supplementary Section~\ref{sec:Storey}, we present an adaptive version of Focused BH that provably controls the FDR under the assumptions of Theorem~\ref{main_theorem}, part (iii).

\subsection{Comparison to other hierarchical testing methods}

Hierarchical structures pose challenges to FDR control but also provide opportunities to gain power. In particular, logical relationships among hypotheses and clustering of non-nulls with respect to the graph structure can both be exploited to boost power. The Structured Holm \cite{MG15} and Yekutieli \cite{Y08} methods exploit the former and latter, respectively. We briefly describe both of these methods and then compare them to Focused BH. 

Structured Holm is designed for FWER control on general DAGs, exploiting logical relationships \eqref{logical_relationships} that are assumed to hold. % among related nodes. 
Like Holm's original procedure, Structured Holm is a step-down improvement of Bonferroni's method, except it also adds ancestors of all rejected hypotheses to the rejection set at each step. Due to the logical relationships, this can be done ``for free" while maintaining FWER control under arbitrary dependence. %Also like the Holm procedure, Structured Holm's FWER guarantee holds under arbitrary dependence. 
Since any FWER-controlling procedure retains the same guarantee after arbitrary filtering, no special adjustments for filtering are necessary.

Yekutieli's hierarchical testing procedure, applicable to trees, is designed to control the FDR while boosting power by adaptively finding subtrees with more non-null nodes. The procedure starts by applying BH at level $q$ to the root nodes of the tree, and then repeatedly applies BH at the same level $q$ to the children of each rejected node from the previous step. Yekutieli proved that, under independent p-values, this procedure controls the FDR among the outer nodes at level $2\cdot D\cdot q \cdot \delta^*$, where $D$ is the depth of the tree and $\delta^*$ is a constant depending on the structure of the problem that is bounded above by 1.44 but is usually near 1. In this paper, we define the Yekutieli procedure by taking $\delta^* = 1$ and applying BH repeatedly, as described above, at level $q/(2\cdot D)$.

Our choice to design Focused BH while allowing for any filter broadens the applicability of our methodology but leaves room for improvements where power can be gained by leveraging the details of the testing problem, as is done by Structured Holm and Yekutieli. Our numerical experiments in Section~\ref{sec:ICD_experiment} confirm that Yekutieli can outperform Focused BH if the signals are highly concentrated in the tree, though Focused BH outperforms the Yekutieli method when signals are more spread out. Indeed, not traversing hypotheses from top to bottom allows Focused BH to avoid the issue that signal may not be very strong near the top of the tree (due to a dilution effect), which poses a problem for all hierarchical methods, including Yekutieli's. It is harder to directly compare Structured Holm to Focused BH because the former controls a more stringent error rate than the latter, but in our simulations we observed that the cost in power of controlling FWER was greater than the gain from exploiting logical relationships. Structured Holm does still have excellent power for a FWER procedure, though this procedure is limited to testing self-contained hypotheses or others that satisfy the logical relationships \eqref{logical_relationships}. Finally, Focused BH is broader than either the Yekutieli procedure or Structured Holm because it applies to non-hierarchically structured multiple testing problems.

%Having described the Focused BH methodology and its theoretical properties, we proceed next to evaluate its performance on simulated and real data.

\section[Experiments]{Numerical simulations} \label{sec:experiments}

In this section, we evaluate the power of Focused BH via numerical simulations. We work in two parallel settings, with graph structures (ICD and GO) and filters (outer nodes and REVIGO \cite{SetS11}) inspired by PheWAS and GO enrichment analysis, respectively. In both cases, we use the following data-generating mechanism. As in Section~\ref{sec:setup_hierarchical}, each node is assigned a set of items. Some items are chosen to be non-null, and each item $k$ receives a p-value from a one-sided test based on $N(\mu_k, 1)$, where $\mu_k = \mu \cdot  \ind(k\text{ non-null})$ for some signal strength $\mu \geq 0$. We use self-contained hypotheses~\eqref{self-contained}, so a node is defined as non-null if it contains any non-null items. We test this hypothesis by aggregating the p-values of the items it contains. While simple, this data-generating mechanism captures p-value dependencies among nodes due to graph relationships. Now, we present each simulation in more detail. We conduct an additional simulation for spatially structured hypotheses (see supplementary Section~\ref{sec:GWAS_simulation}); the results are qualitatively similar to those in the main text.  

\subsection{Simulation with ICD tree} \label{sec:ICD_experiment}

This simulation is based on the real graph of ICD-10 codes, containing $m = 19154$ nodes, with leaves playing the role of items; a leaf node is associated with each of its ancestor nodes. The ICD graph  is broken down into 22 subtrees, each representing a major class of diseases. The graph contains 5 levels, and
Figure~\ref{fig:ICD_simulation_graph_info} shows how many nodes are at each level and how many items these nodes contain. Levels are based on the heights of nodes, defined as the length of the longest path from the node to a leaf of the tree, plus one. Therefore, all leaves are at height 1. We choose non-null nodes using three regimes: \textit{clustered}, \textit{intermediate}, and \textit{dispersed}. In each of the three regimes, we choose 50 leaf nodes to be non-null. We define these by taking the union of the leaf node descendants of a certain set of nodes, which we call \textit{anchor nodes}; for the clustered setting we choose one anchor node at depth 2 with 50 leaf descendants, for the intermediate setting we randomly choose 10 anchor nodes at depth 3 in different subtrees with 5 leaf descendants each, and for the dispersed setting we randomly choose a total of 50 leaf nodes (at depths 4 and 5, where depth is defined as the length of the path from the node to the root node of the corresponding subtree) as anchor nodes. As Figure~\ref{fig:ICD_simulation_graph_info} shows, the more clustered the non-null items, the more the number of non-null nodes decreases as level increases. We aggregate item-level p-values using the Fisher combination test, allowing us to capture the fact that signal may be weak in the leaves, but can strengthen at nodes higher up in the tree when these weak signals are aggregated. 

\begin{figure}[h!]
\centering
\includegraphics[width = 0.9\textwidth]{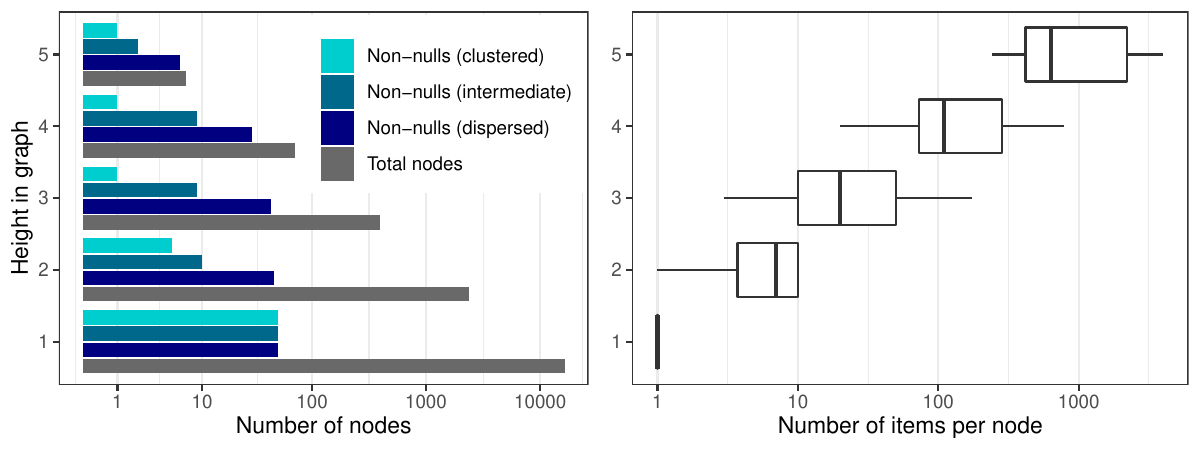}
\caption{Summary of the ICD tree graph used for simulations. Left: numbers of non-null nodes in each of the three configurations and total nodes at each level of the graph. Right: numbers of items per node at each level.}
\label{fig:ICD_simulation_graph_info}
\end{figure}

We compare Focused BH with the outer nodes filter to Yekutieli and Structured Holm, each at  nominal level 0.1, followed by outer node filtering. %the following methods: Structured Holm, Yekutieli, and Focused BH, each at nominal level 0.1, followed by outer node filtering. 
%We must be careful to 
Note that this nominal level has different meaning across methods: for Structured Holm it is the target FWER before filtering, while for Focused BH and Yekutieli it is the target FDR after filtering. Figure~\ref{fig:PheWAS_experiment} shows the power and FDR of these methods as we vary the signal strength. We quantify power as the number of filtered discoveries, divided by the maximum possible number of such discoveries, averaged over 250 repetitions.  %over all methods and signal strengths. 

\begin{figure}[h!]
\centering
\includegraphics[width = 0.85\textwidth]{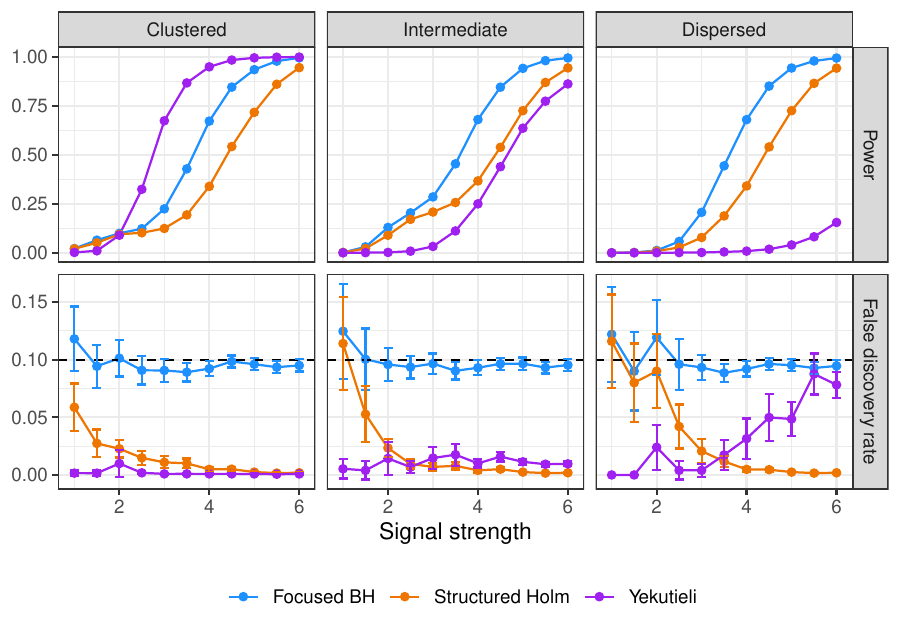}
\caption{Results of PheWAS simulation with outer nodes filter; each point is the average across 250 repetitions. The intervals for FDR have width twice the standard error. The dashed horizontal line shows the target nominal level of 0.1.}
\label{fig:PheWAS_experiment}
\end{figure}

All three methods successfully control the FDR, even though Yekutieli's method assumes independence across the entire tree, and Theorem~\ref{main_theorem} does not guarantee FDR control for Focused BH because it is not clear whether the PRDS assumption holds in this case. %for Fisher combinations of independent p-values. 
Across all settings, Focused BH outperforms Structured Holm; %showing that the power lost by controlling the more stringent FWER is greater than that gained by leveraging logical relationships; 
this is especially true when many discoveries can be made. Finally, we see that the performance of Yekutieli's method depends significantly on the configuration of the non-nulls. It outperforms Focused BH in the clustered setting, but performs worse in the intermediate and dispersed settings. By contrast, Focused BH is relatively more stable to the signal configuration, since the strength of the signal drives its rejection sets more than the graph structure. 

\subsection{Simulation with GO DAG}

This simulation is based on the real DAG of GO biological processes, restricted to the $m = 12409$ terms annotated with at most 100 genes. We make this restriction because the self-contained null hypotheses for GO terms containing too many genes are usually false, but uninterestingly so. Figure~\ref{fig:GO_simulation_graph_info} contains information on the resulting DAG structure, which has several root nodes. We choose non-null nodes in a similar manner to the PheWAS simulation. We choose a set of genes to be non-null, and all nodes containing any non-null genes themselves are non-null. We must be more parsimonious with the number of genes where we plant signals in order to limit the number of non-null nodes, given the more interconnected nature of GO. Therefore, we randomly choose two anchor GO terms among those annotated with five genes, and consider the resulting ten genes non-null. Given the larger number of items per node in the GO DAG, %application, 
we may expect the alternatives within nodes to be sparser and therefore more powerfully detected using the Simes test \cite{S86}. 

\begin{figure}[h!]
\centering
\includegraphics[width = 0.9\textwidth]{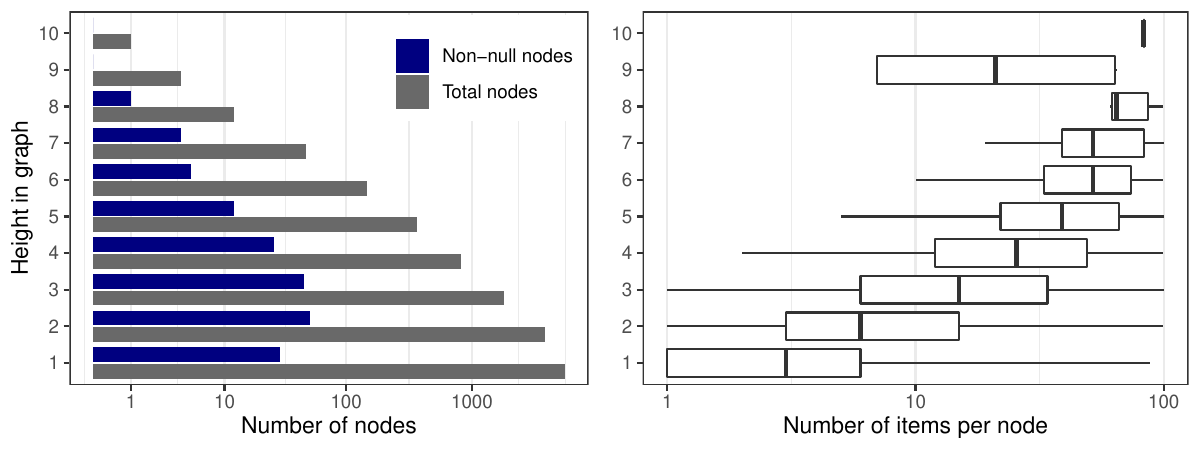}
\caption{Summary of GO subgraph used for simulations; details as in Fig.~\ref{fig:ICD_simulation_graph_info}. }
\label{fig:GO_simulation_graph_info}
\end{figure}

We exclude Yekutieli from the GO simulation because it does not apply to non-tree DAGs. On the other hand, we add the resampling-based version of Focused BH introduced in Section~\ref{sec:improving}. We generate the estimate $\widehat V$ by resampling from the global null distribution, which in this case is known exactly. Figure~\ref{fig:permutation_estimates} shows these estimates $\widehat V$ for the GO and ICD graphs (normalized by their respective numbers of nodes for direct comparison) as well as the original linear estimate. We can see that the resampling estimate is further away from the linear estimate for the GO graph with REVIGO than for the ICD graph with outer nodes filter. This fact, due in part to the fact that REVIGO generally reduces the rejection set by a larger factor, suggests the resampled Focused BH will make more of a power improvement in this simulation and is why we excluded it from the previous one. 

Figure~\ref{fig:REVIGO_experiment} shows the results of the simulation. We see some of the same trends as in Figure~\ref{fig:PheWAS_experiment}, such as the conservative behavior of Structured Holm. %Also note that 
The resampled version of Focused BH successfully controls the FDR across all signal strengths and improves on the power of the original method, the latter as suggested by Figure~\ref{fig:permutation_estimates}. Note that the original version of Focused BH controls the FDR, though it is unclear if REVIGO is monotonic.

\begin{figure}[h!]
\centering
\includegraphics[width = 0.75\textwidth]{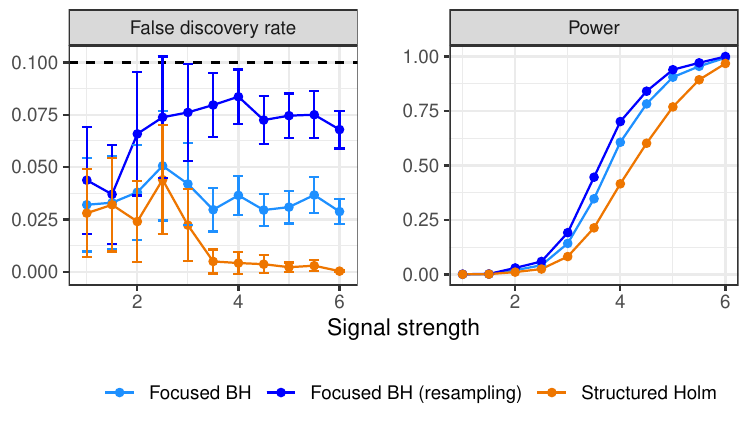}
\caption{Results of GO simulation with REVIGO filter; details as in Figure~\ref{fig:PheWAS_experiment}.}
\label{fig:REVIGO_experiment}
\end{figure}

In summary, these two numerical experiments show that Focused BH and its resampling-based improvement enjoy good power across a variety of simulation settings, usually outperforming existing alternatives. This is not universally true, however, as we saw in the clustered setting in Figure~\ref{fig:PheWAS_experiment}, where the Yekutieli method outperforms the others. These experiments also demonstrate that the FDR control of Focused BH extends beyond the sufficient conditions we provide in Theorem~\ref{main_theorem}, echoing a property of the BH procedure. 

Finally, Figure~\ref{fig:BH_experiments} compares the performance of Focused BH to that of naive non-hierarchical methods discussed in the introduction in the same simulation settings as above. In particular, we apply BH to all nodes at level $q$ followed by a filtering step. Additionally, for the PheWAS simulation, we apply BH at level $q$ to only the leaf nodes, only the nodes one level above leaf nodes, and only the nodes two levels above leaf nodes. Note that the level-specific methods are guaranteed to control the FDR after outer nodes filtering, since this filter leaves the rejection set unchanged. As anticipated, both filters cause BH to lose FDR control, reinforcing our message that it is insufficient to apply traditional FDR procedures in the presence of filtering. For the PheWAS simulation, the power of Focused BH essentially matches that of BH, so it effectively prunes out false positives without removing many true positives. In the GO simulation, there is a greater gap between the power of Focused BH and BH, but this gap is narrowed by the resampling estimate. While level-specific BH does control the FDR, Figure~\ref{fig:BH_experiments} shows that it can have very low power if the level is chosen poorly, while Focused BH chooses the levels of its discoveries adaptively. 

%Having demonstrated the performance of Focused BH in simulations, we move on to applying this method to real PheWAS and GO enrichment analysis data sets. 

\section{Applications} \label{sec:data_analysis}

\subsection{UK Biobank PheWAS} \label{sec:phewas}

We perform a phenome-wide association analysis on the UK Biobank data \cite{BetJ18}, which contains genotype and electronic health record information for 502,616 UK residents. A set of ICD-10 disease codes has been extracted from each individual's health record. The genetic variant of interest is HLA-B*27:05, which codes for a sub-type of the human leukocyte antigen called HLA-B*27. HLA-B*27:05 is famous for its association with ankylosing spondylitis, and is known to be associated with many other diseases as well. A PheWAS for this variant was conducted previously by \cite{CetM17},  though it is not directly comparable with Focused BH because it relies on a Bayesian analysis framework.

To obtain a p-value for each ICD-10 code, we perform a chi-squared test of independence between the disease and the genotype variable (which takes values in $\{0,1,2\}$). Note that this is not a self-contained test, and the logical relationships~(\ref{logical_relationships}) no longer hold. Furthermore, these hypotheses are generally harder to reject than self-contained ones. Many of the ICD-10 terms correspond to diseases that are rare, and there is little power to reject them. Thus, we first restricted our attention to the $m = 3265$ diseases observed at least 50 times in the UK Biobank data set (see Figure~\ref{fig:ICD_analysis_graph_info}). This filtering step is benign and need not be accounted for because it does not involve the response variable (the genotype). 

We compare the outer node discoveries of the following  methods: % at a nominal level of $q=0.05$:  
Structured Holm (SH), Yekutieli (Y), Focused BH (FBH), BH, Leaf BH (LBH), and Holm (H). %and compare them in terms of their 
%outer node discoveries. 
The first three are those compared in Section \ref{sec:ICD_experiment}. %the simulations from the previous section. 
Leaf BH is the naive method where we apply BH only to the leaf nodes of the ICD-10 tree (after having removed the rare diseases). One might hope that we have decent power to find associations with diseases with at least 50 cases. Note that Structured Holm may not be valid due to its reliance on logical relationships, so we also apply the Holm method for FWER control, which does not rely on this assumption. We apply each method at a nominal level of 0.05.% $q = 0.05$. 

Figure~\ref{fig:UKBB_overview} summarizes the number of findings of each method (left) and their overlaps (right). Holm has the same rejections as Structured Holm. Aside from BH, which we do not expect to control the outer nodes FDR, Focused BH has the most findings (24). All but one of the discoveries of Leaf BH, Structured Holm, or Yekutieli were also discovered by Focused BH, while four of the Focused BH discoveries were not found by any of these methods, as we see from the Venn diagram in this figure. While the proper validation of these discoveries is beyond the scope of this work, we present some evidence in the medical literature for these associations in Section \ref{sec:UKBB_discoveries} of the supplement. Recall that the ICD is broken down into 22 subtrees, each representing a major class of diseases. Figure~\ref{fig:UKBB_overview} also shows that Focused BH finds associations with HLA-B*27:05 in 10 of these 22 disease classes, demonstrating that this variant has effects on a variety of diseases. On the other hand, Leaf BH, Structured Holm, and Yekutieli find effects in only 8,6, and 1 disease classes, respectively. We can explain the low power of the Yekutieli method by recalling from the previous section that it performs better in regimes where signals are clustered together. In this data, clearly the signal is dispersed throughout the subtrees, see the table in Figure~\ref{fig:UKBB_overview}.

Figure~\ref{fig:UKBB_MSK} gives a closer look at the findings in the subtree of musculoskeletal diseases. In this subtree, Focused BH finds seven associations. Four of these are also discovered by Structured Holm and Leaf BH, one is discovered by Structured Holm but not Leaf BH, one is discovered by Leaf BH but not Structured Holm, and one is not discovered by either. The two not discovered by Structured Holm are the weakest two signals, a consequence of this method's more stringent error control. The two not discovered by Leaf BH are not leaf nodes; indeed, outer node discoveries can be internal nodes of the set of tested hypotheses. 

\begin{figure}[h!]
\centering
$\vcenter{
\hbox{
	\includegraphics[width = 0.5\textwidth]{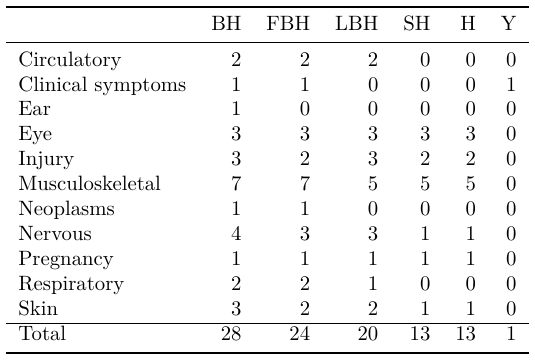}		
}}$
$\vcenter{
	\hbox{
		\includegraphics[width = 0.35\textwidth]{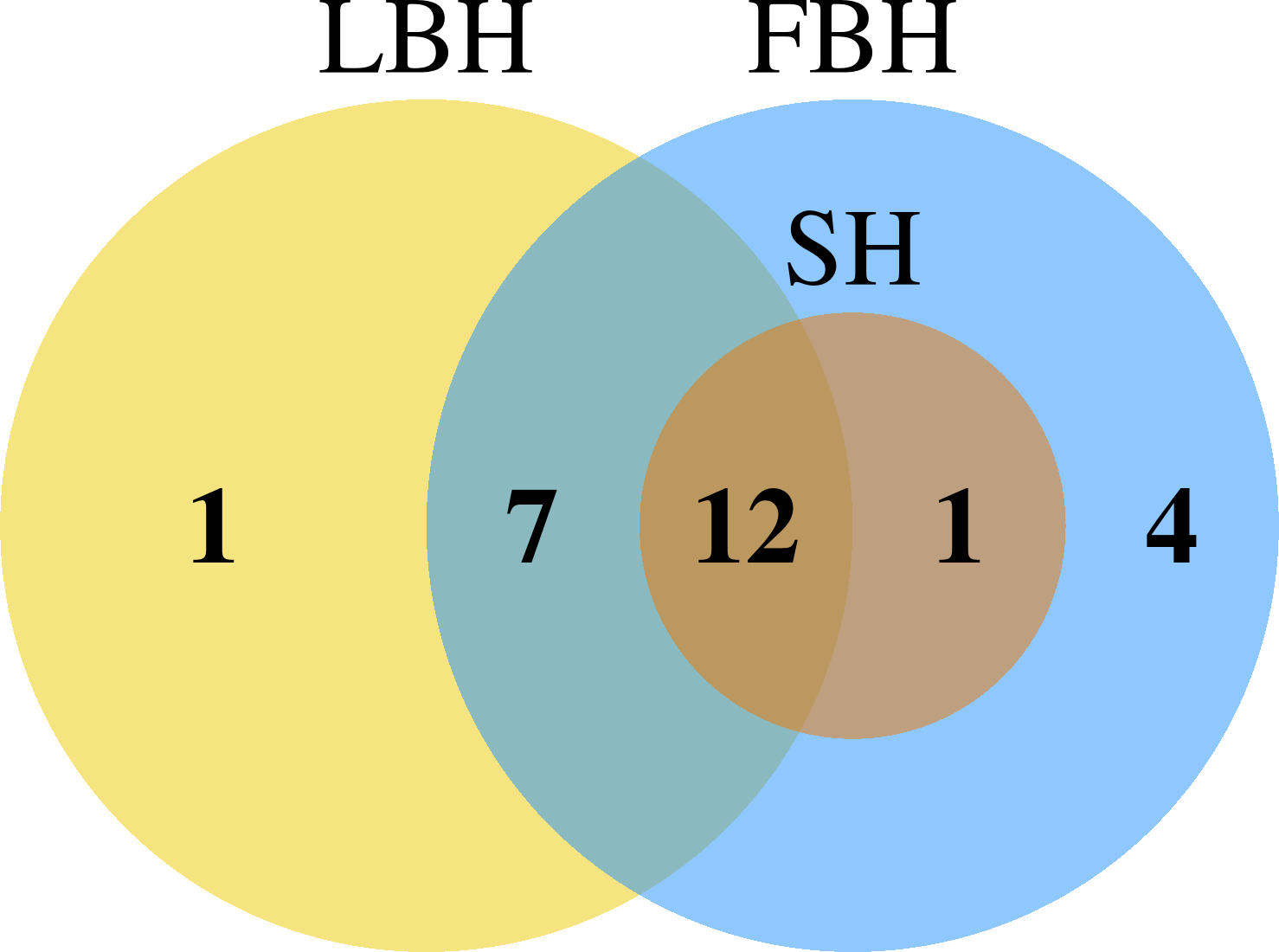}		
	}}$
	\caption{Associations found in PheWAS analysis; the table (left) shows the breakdown by disease category and the Venn diagram (right) shows intersections among all discoveries of three methods.}
	\label{fig:UKBB_overview}
\end{figure}

\begin{figure}[h!]
	\centering
	\includegraphics[width = 0.55\textwidth, valign=c]{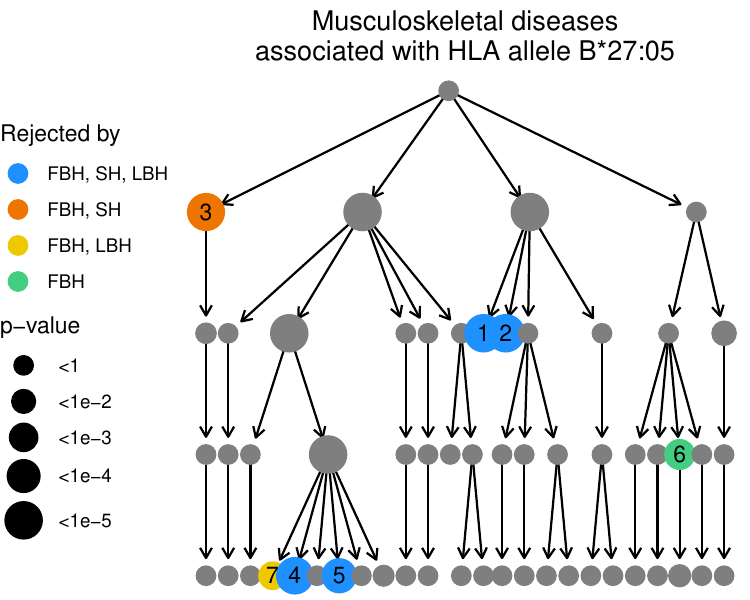}
	\includegraphics[width = 0.42\textwidth, valign=c]{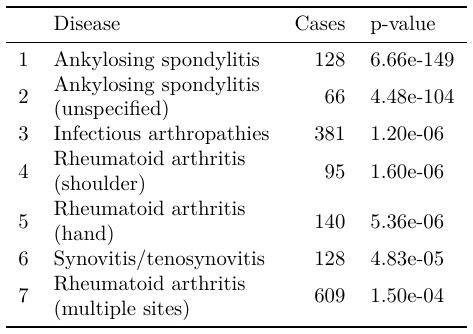}
	\caption{Musculoskeletal diseases found in PheWAS analysis. Associations are shown in the context of the ICD-10 graph (left), and listed individually (right). Non-discovered nodes are subsampled for visualization purposes.}
	\label{fig:UKBB_MSK}
\end{figure}

\subsection{GO enrichment analysis}

Next, we analyze data from a gene expression experiment to study the difference between breast cancer patients who remained cancer-free for five years after treatment and those who did not \cite{VetW02}. By cross-referencing the differentially expressed genes with biological processes, the goal is to discover processes contributing to breast cancer treatment outcome.

For this analysis, we used two of the most common tools in GO enrichment analysis: GOrilla \cite{EetY09} to compute enrichment p-values from a list of differentially expressed genes and REVIGO \cite{SetS11} for filtering. In fact, GOrilla actually links to REVIGO in order to help users filter the lists of GO terms it outputs. In more detail, we started with a list of $K = 9113$ genes (ordered according to their differential expression in the breast cancer data set) that is available on the GOrilla website, where this data is used as the ``running example." We applied GOrilla with default settings to run the enrichment analysis for this ordered gene list (using the mHG statistic \cite{EetY09}) on $m = 14016$ terms from the ``Biological Process" sub-ontology of GO (see Figure~\ref{fig:REVIGO_analysis_graph_info}). GOrilla runs what \cite{GB07} call a \textit{competitive test}, to see if the genes in a GO term are distributed nearer the top of the ordering by differential expression than expected by random chance. Competitive tests are more often used in practice but unlike self-contained tests, do not imply logical relationships. This yielded enrichment p-values for each GO term, which we downloaded from GOrilla. We then applied BH, Focused BH, Structured Holm, and Holm to these enrichment p-values, each with nominal level $q = 0.1$.
\begin{table}[h!]
	\centering
	\includegraphics[width = 0.7\textwidth]{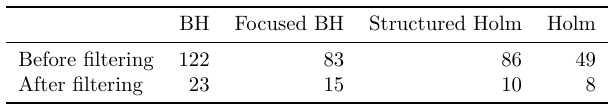}
	\caption{Numbers of discoveries for GO enrichment data, before and after filtering.}
	\label{tab:REVIGO}
\end{table}

Table~\ref{tab:REVIGO} shows the numbers of discoveries made by each of these three methods. BH made the most filtered discoveries (23), though Figure~\ref{fig:BH_experiments} suggests these may not be trustworthy. Focused BH makes 15 discoveries, while Structured Holm makes 10. Interestingly, Structured Holm actually makes more discoveries than Focused BH before filtering. This occurs because Structured Holm augments the rejection set with all of its ancestors, but most of these are redundant and therefore filtered out by REVIGO. 

\begin{figure}[h!]
	\centering
	\includegraphics[width = 0.55\textwidth, valign=c]{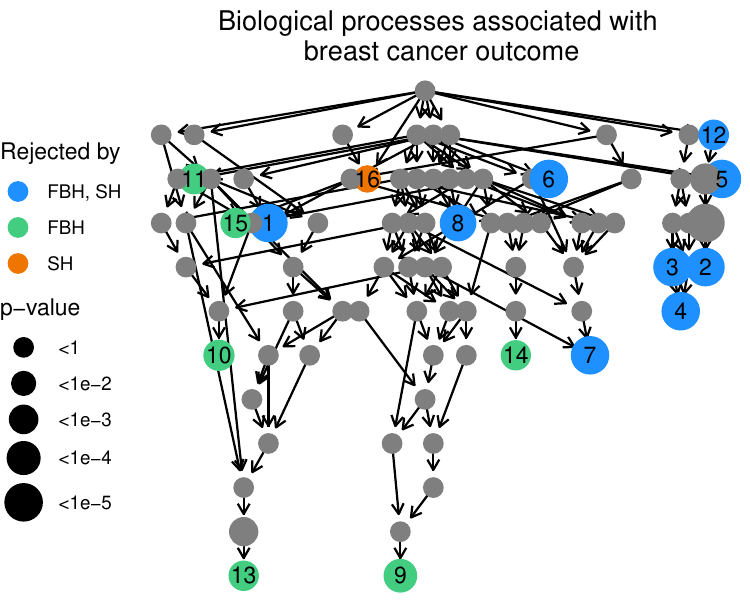}
	\includegraphics[width = 0.42\textwidth, valign=c]{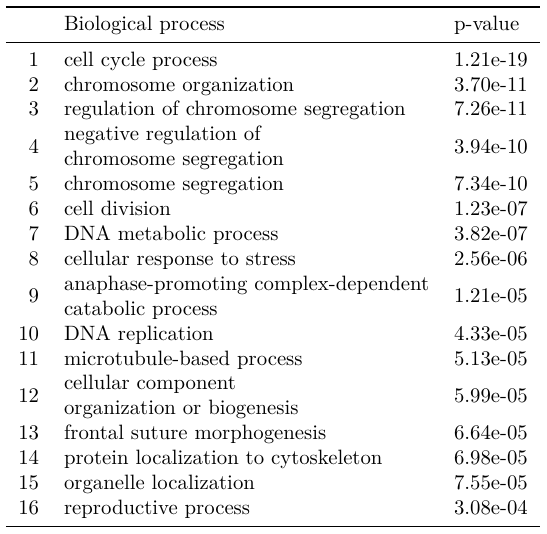}
	\caption{GO enrichment analysis results.}.
	\label{fig:GO}
\end{figure}
Figure~\ref{fig:GO} shows the 16 terms discovered by Focused BH or Structured Holm in their graph context, and lists them along with their p-values. We see that the rejected terms are reasonably separated in the graph thanks to the filter. The terms rejected by Focused BH but not Structured Holm generally have weaker signals, as we saw in the PheWAS data analysis. As an exception, Structured Holm finds an association with the term ``reproductive process" (16) that Focused BH does not, even though it has a relatively weak p-value. This is due to Structured Holm's leveraging of logical relationships, allowing nodes with weaker p-values to borrow power from nodes they are connected to with stronger p-values. 

\section{Discussion} \label{sec:discussion}

We have introduced Focused BH, a method accommodating a variety of pre-specified filters while guaranteeing FDR control. Our %numerical 
simulations in Section~\ref{sec:experiments} showed that Focused BH performs quite well in comparison to available methodologies. The method is also very flexible, allowing easy integration with filters arising in a variety of application domains. We saw this in Section~\ref{sec:data_analysis}, where we applied Focused BH seamlessly to two real applications with different graph structures and different filters. It is likely to be applicable to other problems involving hierarchically structured hypotheses, %
like microbiome analysis, where bacterial species are organized into phylogenetic trees, or variable selection, which has been posed as a hierarchical multiple testing problem \cite{M08}, with nodes representing correlated groups of variables of different sizes. 
%like microbiome analysis, where bacterial species are organized into phylogenetic trees. Another problem where Focused BH may prove useful is variable selection, which has been posed as a hierarchical multiple testing problem \cite{M08}, with nodes representing correlated groups of variables of different sizes. 
Our theoretical guarantees cover a reasonably broad range of filters and our simulations underscore the robustness of our method: users can expect FDR control even in certain settings which are beyond our theoretical results. 

Focused BH, though already very general, can be extended in various ways. For example, filtering operations can extend beyond subsetting. We may also consider \textit{prioritization} operations, where elements of $\cR^*$ obtain ``prioritization weights" reflecting the relative importance of these discoveries. For example, more interesting leads might be followed up with more resources like time or money. Furthermore, the relative importance of discoveries may be more accurately assessed \textit{after} looking at the rejection set. Therefore, in supplementary Section~\ref{sec:app_fractional}, we introduce a more general definition of filtering that encompasses any such subsetting and prioritization operations. Focused BH and its theoretical guarantees easily generalize to this broader class of filters. This puts Focused BH in the realm of weighted hypothesis testing, but unlike previous works on this subject, the hypothesis weights need not be fixed ahead of time and may be determined by the data. In addition, Focused BH can be extended to adapt to the non-null proportion, to handle arbitrarily dependent p-values, or to handle multiple filters; see supplementary Section~\ref{sec:extensions}. %Focused BH also has implications beyond hierarchical testing, which are the subject of Section~\ref{sec:implications_beyond}. We discuss future research directions in Section~\ref{sec:future_work}.

\subsection{Implications beyond hierarchical testing and future work} \label{sec:implications_beyond}

While hierarchical applications motivated our development of Focused BH, and are among its most promising applications, 
the generality of our methodology leads to implications beyond hierarchical multiple testing. Other kinds of structure present in testing problems can be tackled with Focused BH as well. 
For example, in applications with spatial structure, rejections too near each other might be considered redundant. This is the case in Genome-Wide Association Studies (GWAS), which test the associations between one phenotype and millions of genetic variants. Rather than reporting all discoveries, scientists identify clusters of rejected variants residing in the same genomic region and report only the variant among these with the smallest p-value. This is also a filtering operation known to inflate FDR \cite{SetY11,BetS17a}. In supplementary Section~\ref{sec:GWAS_simulation}, we demonstrate via simulation how Focused BH restores FDR control for GWAS. Similar issues arise in imaging applications \cite{PetW04, BH07, SetS15}. 

More generally, Focused BH can be viewed as a means to control the FDR while enforcing arbitrary structural constraints. In fact, there are close connections between filtering and structural constraints, if we view filters as ``projections" of arbitrary rejection sets onto those satisfying structural constraints; see supplementary Section~\ref{sec:supp_structured} for more discussion.

Focused BH is also a form of \textit{selective inference}, since it provides valid inference despite a selection step. Some such methodologies \cite{BB14, Taylor2015, BetS17a} conduct inference on the set of hypotheses surviving a pre hoc screening procedure; others account for post hoc selection rules through simultaneous inference \cite{GS11mt, BetZ13, Katsevich2020, GillesBlanchard2019}. Focused BH is somewhere in between: the selection (i.e. the filter) is applied to the output of a multiple testing procedure rather than to its input, but it must be specified in advance. Focused BH reduces to the method of \cite{BetS17a} when $\F$ is a screening filter, but goes beyond it by allowing filters that operate on rejection sets, like the outer nodes filter or REVIGO.

%\subsection{Future work} \label{sec:future_work}

There are several directions for extending the work presented in this paper. First, our numerical simulations suggest that the resampling-based version of Focused BH controls the FDR while improving power for filters that reduce the number of rejected nulls. Therefore, developing theory for this methodology would be a worthwhile endeavor; supplementary Section~\ref{sec:supp_perm} may provide a useful starting point. Second, it would be interesting to see if the advantages of Focused BH can be combined with the abilities of Structured Holm and Yekutieli to leverage logical relationships and non-uniform signal distributions, respectively, 
to get a more specialized procedure with higher power. Finally, it would be exciting to apply Focused BH to problems in other applied domains, hierarchical or otherwise.

\section{Acknowledgements}
We thank V. Griskin, J. Zhu, E. Candes, M. Sesia, A. Weinstein, A. Ramdas, D. Yekutieli, M. Celentano, and A. Katsevich for helpful discussions. We also thank G. Bejerano and Y. Tanigawa for their help with the application to GO enrichment analysis  and A. Solari for his comments on an earlier version of the manuscript.

\bigskip
\begin{center}
	{\large\bf SUPPLEMENTARY MATERIAL}
\end{center}

\begin{description}
	
	\item[Supplementary text:] Supplementary discussion and proofs of all results. (.pdf file)
	
\end{description}

\bibliographystyle{unsrt}

%\bibliography{Mybib}

\begin{thebibliography}{10}
	
	\bibitem{AetE00}
	Michael Ashburner, Catherine~A Ball, Judith~A Blake, David Botstein, Heather
	Butler, J~Michael Cherry, Allan~P Davis, Kara Dolinski, Selina~S Dwight, and
	Janan~T Eppig.
	\newblock {Gene Ontology: tool for the unification of biology}.
	\newblock {\em Nature Genetics}, 25(1):25, 2000.
	
	\bibitem{DetC10}
	Joshua~C. Denny, Marylyn~D. Ritchie, Melissa~A. Basford, Jill~M. Pulley, Lisa
	Bastarache, Kristin Brown-Gentry, Deede Wang, Dan~R. Masys, Dan~M. Roden, and
	Dana~C. Crawford.
	\newblock {PheWAS: demonstrating the feasibility of a phenome-wide scan to
		discover gene--disease associations}.
	\newblock {\em Bioinformatics}, 26(9):1205--1210, 2010.
	
	\bibitem{BH95}
	Yoav Benjamini and Yosef Hochberg.
	\newblock {Controlling the false discovery rate: a practical and powerful
		approach to multiple testing}.
	\newblock {\em Journal of the Royal Statistical Society: Series B (Statistical
		Methodology)}, 57(1):289--300, 1995.
	
	\bibitem{Y08}
	Daniel Yekutieli.
	\newblock {Hierarchical false discovery rate-controlling methodology}.
	\newblock {\em Journal of the American Statistical Association},
	103(481):309--316, 2008.
	
	\bibitem{GS11}
	Yongtao Guan and Matthew Stephens.
	\newblock {Bayesian variable selection regression for genome-wide association
		studies and other large scale problems}.
	\newblock {\em The Annals of Applied Statistics}, 5(3):1780--1815, 2011.
	
	\bibitem{MG15_Holm}
	Rosa~J. Meijer and Jelle~J. Goeman.
	\newblock {A multiple testing method for hypotheses structured in a directed
		acyclic graph}.
	\newblock {\em Biometrical Journal}, 57(1):123--143, 2015.
	
	\bibitem{LG16}
	Gavin Lynch and Wenge Guo.
	\newblock {On Procedures Controlling the FDR for Testing Hierarchically Ordered
		Hypotheses}.
	\newblock {\em arXiv}, 2016.
	
	\bibitem{FBR15}
	Rina {Foygel Barber} and Aaditya Ramdas.
	\newblock {The p-filter: multi-layer false discovery rate control for grouped
		hypotheses}.
	\newblock {\em Journal of the Royal Statistical Society: Series B (Statistical
		Methodology)}, 79(4):1247--1268, 2016.
	
	\bibitem{DAGGER}
	Aaditya Ramdas, Jianbo Chen, Martin~J. Wainwright, and Michael~I. Jordan.
	\newblock {DAGGER: A sequential algorithm for FDR control on DAGs}.
	\newblock {\em Biometrika}, 106(1):69--86, 2019.
	
	\bibitem{BetS17a}
	Damian Brzyski, Christine~B. Peterson, Piotr Sobczyk, Emmanuel~J. Candes,
	Malgorzata Bogdan, and Chiara Sabatti.
	\newblock {Controlling the rate of GWAS false discoveries}.
	\newblock {\em Genetics}, 205(1):61--75, 2017.
	
	\bibitem{CetM17}
	Adrian Cortes, Calliope~A. Dendrou, Allan Motyer, Luke Jostins, Damjan
	Vukcevic, Alexander Dilthey, Peter Donnelly, Stephen Leslie, Lars Fugger, and
	Gil McVean.
	\newblock {Bayesian analysis of genetic association across tree-structured
		routine healthcare data in the UK Biobank}.
	\newblock {\em Nature Genetics}, 49(9):1311--1318, 2017.
	
	\bibitem{katsevich17mkf}
	Eugene Katsevich and Chiara Sabatti.
	\newblock {Multilayer knockoff filter: Controlled variable selection at
		multiple resolutions}.
	\newblock {\em The Annals of Applied Statistics}, 13(1):1--33, 2019.
	
	\bibitem{M08}
	Nicolai Meinshausen.
	\newblock {Hierarchical testing of variable importance}.
	\newblock {\em Biometrika}, 95(2):265--278, 2008.
	
	\bibitem{MG15}
	Rosa~J. Meijer and Jelle~J. Goeman.
	\newblock {Multiple testing of gene sets from Gene Ontology: possibilities and
		pitfalls}.
	\newblock {\em Briefings in Bioinformatics}, 17(5):808--818, 2015.
	
	\bibitem{GB07}
	Jelle Goeman and Peter Buhlmann.
	\newblock {Analyzing gene expression data in terms of gene sets: methodological
		issues}.
	\newblock {\em Bioinformatics}, 23(8):980--987, 2007.
	
	\bibitem{SetS11}
	Fran Supek, Matko Bo{\v{s}}njak, Nives {\v{S}}kunca, and Tomislav {\v{S}}muc.
	\newblock {REVIGO summarizes and visualizes long lists of gene ontology terms}.
	\newblock {\em PloS one}, 6(7):e21800, 2011.
	
	\bibitem{BB14}
	Yoav Benjamini and Marina Bogomolov.
	\newblock {Selective inference on multiple families of hypotheses}.
	\newblock {\em Journal of the Royal Statistical Society: Series B (Statistical
		Methodology)}, 76(1):297--318, 2014.
	
	\bibitem{Storey04}
	John~D. Storey, Jonathan~E. Taylor, and David Siegmund.
	\newblock {Strong control, conservative point estimation and simultaneous
		conservative consistency of false discovery rates: a unified approach}.
	\newblock {\em Journal of the Royal Statistical Society: Series B (Statistical
		Methodology)}, 66(1):187--205, 2004.
	
	\bibitem{BY01}
	Yoav Benjamini and Daniel Yekutieli.
	\newblock {The control of the false discovery rate in multiple testing under
		dependency}.
	\newblock {\em The Annals of Statistics}, 29(4):1165--1188, 2001.
	
	\bibitem{S86}
	R.~J. Simes.
	\newblock {An improved Bonferroni procedure for multiple tests of
		significance}.
	\newblock {\em Biometrika}, 73(3):751--754, 1986.
	
	\bibitem{BH13}
	Marina Bogomolov and Ruth Heller.
	\newblock {Discovering findings that replicate from a primary study of high
		dimension to a follow-up study}.
	\newblock {\em Journal of the American Statistical Association},
	108(504):1480--1492, 2013.
	
	\bibitem{BR08}
	Gilles Blanchard and Etienne Roquain.
	\newblock {Two simple sufficient conditions for FDR control}.
	\newblock {\em Electronic Journal of Statistics}, 2:963--992, 2008.
	
	\bibitem{RetJ17}
	Aaditya Ramdas, Rina {Foygel Barber}, Martin~J. Wainwright, and Michael~I.
	Jordan.
	\newblock {A unified treatment of multiple testing with prior knowledge using
		the p-filter}.
	\newblock {\em The Annals of Statistics}, 47(5):2790--2821, 2019.
	
	\bibitem{YB99}
	Daniel Yekutieli and Yoav Benjamini.
	\newblock {Resampling-based false discovery rate controlling multiple test
		procedures for correlated test statistics}.
	\newblock {\em Journal of Statistical Planning and Inference},
	82(1-2):171--196, 1999.
	
	\bibitem{TC01}
	Virginia~G. Tusher, Robert Tibshirani, and Gilbert Chu.
	\newblock {Significance analysis of microarrays applied to the ionizing
		radiation response}.
	\newblock {\em Proceedings of the National Academy of Sciences},
	98(9):5116--5121, 2001.
	
	\bibitem{HG18}
	Jesse Hemerik and Jelle~J. Goeman.
	\newblock {False discovery proportion estimation by permutations: confidence
		for significance analysis of microarrays}.
	\newblock {\em Journal of the Royal Statistical Society: Series B (Statistical
		Methodology)}, 80(1):137--155, 2018.
	
	\bibitem{BetJ18}
	Clare Bycroft, Colin Freeman, Desislava Petkova, Gavin Band, Lloyd~T. Elliott,
	Kevin Sharp, Allan Motyer, Damjan Vukcevic, Olivier Delaneau, Jared
	O'Connell, Adrian Cortes, Samantha Welsh, Alan Young, Mark Effingham, Gil
	McVean, Stephen Leslie, Naomi Allen, Peter Donnelly, and Jonathan Marchini.
	\newblock {The UK Biobank resource with deep phenotyping and genomic data}.
	\newblock {\em Nature}, 562(7726):203, 2018.
	
	\bibitem{VetW02}
	Laura~J. {Van't Veer}, Hongyue Dai, Marc~J. {Van De Vijver}, Yudong~D. He,
	Augustinus A.~M. Hart, Mao Mao, Hans~L. Peterse, Karin {Van Der Kooy},
	Matthew~J. Marton, and Anke~T. Witteveen.
	\newblock {Gene expression profiling predicts clinical outcome of breast
		cancer}.
	\newblock {\em Nature}, 415(6871):530, 2002.
	
	\bibitem{EetY09}
	Zohar Yakhini, Eran Eden, Roy Navon, Israel Steinfeld, and Doron Lipson.
	\newblock {GOrilla: a tool for discovery and visualization of enriched GO terms
		in ranked gene lists}.
	\newblock {\em BMC Bioinformatics}, 10:48, 2009.
	
	\bibitem{SetY11}
	David~O. Siegmund, Nancy~R. Zhang, and Benjamin Yakir.
	\newblock {False discovery rate for scanning statistics}.
	\newblock {\em Biometrika}, 98(4):979--985, 2011.
	
	\bibitem{PetW04}
	Marco {Perone Pacifico}, Christopher Genovese, Isabella Verdinelli, and Larry
	Wasserman.
	\newblock {False discovery control for random fields}.
	\newblock {\em Journal of the American Statistical Association},
	99(468):1002--1014, 2004.
	
	\bibitem{BH07}
	Yoav Benjamini and Ruth Heller.
	\newblock {False discovery rates for spatial signals}.
	\newblock {\em Journal of the American Statistical Association},
	102(480):1272--1281, 2007.
	
	\bibitem{SetS15}
	Wenguang Sun, Brian~J. Reich, T.~Tony Cai, Michele Guindani, and Armin
	Schwartzman.
	\newblock {False Discovery Control in Large-Scale Spatial Multiple Testing}.
	\newblock {\em Journal of the Royal Statistical Society: Series B (Statistical
		Methodology)}, 77(1):59--83, 2015.
	
	\bibitem{Taylor2015}
	Jonathan Taylor, Robert~J. Tibshirani, Rollin Brant, and John~D. Storey.
	\newblock {Statistical learning and selective inference}.
	\newblock {\em Proceedings of the National Academy of Sciences},
	112(25):7629--7634, 2015.
	
	\bibitem{GS11mt}
	Jelle Goeman and Aldo Solari.
	\newblock {Multiple testing for exploratory research}.
	\newblock {\em Statistical Science}, 26(4):584--597, 2011.
	
	\bibitem{BetZ13}
	Richard Berk, Lawrence Brown, Andreas Buja, Kai Zhang, and Linda Zhao.
	\newblock {Valid post-selection inference}.
	\newblock {\em The Annals of Statistics}, 41(2):802--837, 2013.
	
	\bibitem{Katsevich2020}
	Eugene Katsevich and Aaditya Ramdas.
	\newblock {Simultaneous high-probability bounds on the false discovery
		proportion in structured, regression, and online settings}.
	\newblock {\em The Annals of Statistics, to appear}, 2020.
	
	\bibitem{GillesBlanchard2019}
	Gilles Blanchard, Pierre Neuvial, and Etienne Roquain.
	\newblock {Post hoc confidence bounds on false positives using reference
		families}.
	\newblock {\em The Annals of Statistics, to appear}, 2020.
	
	\bibitem{BR07}
	Gilles Blanchard and {\'{E}}tienne Roquain.
	\newblock {Adaptive False Discovery Rate Control under Independence and
		Dependence}.
	\newblock {\em Journal of Machine Learning Research}, 10(Dec):2837--2871, 2009.
	
	\bibitem{YetY06}
	Yoav Benjamini, Abba~M. Krieger, and Daniel Yekutieli.
	\newblock {Adaptive linear step-up procedures that control the false discovery
		rate}.
	\newblock {\em Biometrika}, 93(3):491--507, 2006.
	
	\bibitem{WY93}
	Peter Westfall and S.~Stanley Young.
	\newblock {\em {Resampling-Based Multiple Testing: Examples and Methods for
			p-Value Adjustment}}.
	\newblock Wiley, first edition, 1993.
	
	\bibitem{RY13}
	Grzegorz~A. Rempala and Yuhong Yang.
	\newblock {On permutation procedures for strong control in multiple testing
		with gene expression data}.
	\newblock {\em Statistics and its interface}, 6(1), 2013.
	
	\bibitem{Storey02}
	John Storey.
	\newblock {A direct approach to false discovery rates}.
	\newblock {\em Journal of the Royal Statistical Society: Series B (Statistical
		Methodology)}, 64(3):479--498, 2002.
	
	\bibitem{Colmegna2004}
	Ines Colmegna, Raquel Cuchacovich, and Luis~R. Espinoza.
	\newblock {HLA-B27-Associated Reactive Arthritis: Pathogenetic and Clinical
		Considerations}.
	\newblock {\em Clinical Microbiology Reviews}, 17(2):348--369, 2004.
	
	\bibitem{Ngaruiya2013}
	Christine~M. Ngaruiya and Ian~B.K. Martin.
	\newblock {A case of reactive arthritis: A great masquerader}.
	\newblock {\em American Journal of Emergency Medicine}, 31(1):266.e5--266.e7,
	2013.
	
	\bibitem{Machulla2003}
	Helmut~K.G. Machulla, Frank Steinborn, Michail Tschigrjai, J{\"{u}}rgen
	Langner, and Nikolai~G. Rainov.
	\newblock {Meningioma: Is There an Association with Human Leukocyte Antigens?}
	\newblock {\em Cancer Epidemiology Biomarkers and Prevention},
	12(12):1438--1442, 2003.
	
	\bibitem{SWISS}
	R~Welch.
	\newblock {{\{}$\backslash$tt SWISS{\}}: Software to help identify overlap
		between association scan results and GWAS hit catalogs}, 2014.
	
	\bibitem{BH97}
	Yoav Benjamini and Yosef Hochberg.
	\newblock {Multiple hypotheses testing with weights}.
	\newblock {\em Scandinavian Journal of Statistics}, 24(3):407--418, 1997.
	
	\bibitem{PetC09}
	Catia Pesquita, Daniel Faria, Andr{\'{e}}~O. Falc{\~{a}}o, Phillip Lord, and
	Francisco~M. Couto.
	\newblock {Semantic similarity in biomedical ontologies}.
	\newblock {\em PLoS Computational Biology}, 5(7):e1000443, 2009.
	
	\bibitem{LetF17}
	Gavin Lynch, Wenge Guo, Sanat~K Sarkar, and Helmut Finner.
	\newblock {The control of the false discovery rate in fixed sequence multiple
		testing}.
	\newblock {\em Electronic Journal of Statistics}, 11(2):4649--4673, 2017.
	
	\bibitem{LB17}
	Ang Li and Rina~Foygel Barber.
	\newblock {Accumulation tests for FDR control in ordered hypothesis testing}.
	\newblock {\em Journal of the American Statistical Association},
	112(518):837--849, 2017.
	
\end{thebibliography}

\appendix

\section{Proofs of main results} \label{sec:supp_proofs}

\subsection{FDR control lemma}

In this section, we state a useful lemma, inspired by \cite{BR08, BR07}, from which we can deduce all of our theoretical results. It brings together several known results, and also contains a novel extension. First, we introduce a few definitions.

\begin{definition}[\cite{BR07}] \label{def:thresh_collection}
	Let $\beta: \br^+ \rightarrow \br^+$ be a non-decreasing function, called the shape function, let $G:[0,1]^m \rightarrow (0, \infty)$ be an estimator of the quantity $\frac{m}{m_0}$, and let $q \in (0,1)$ be a target FDR level. Then, the function
	\begin{equation}
	\Delta(\bm p, r) = \frac{q}{m}\beta(r)G(\bm p)
	\label{thresh_collection}
	\end{equation}
	is an adaptive threshold collection with respect to $\beta, G, q, m$.
\end{definition}

To prove FDR statements involving the adaptive estimator $G$, \cite{YetY06, BR07} require the following property:
\begin{equation}
\mathbb E[G(\bm p_{0, j})] \leq \frac{m}{m_0} \  \text{ for all } j \in \mathcal H_0, \text{ where }\ \bm p_{0,j} \equiv (p_1, \dots, p_{j-1}, 0, p_{j+1}, \dots, p_m).
\label{m_0_hat_criterion}
\end{equation}
In other words, $G$ must be a conservative estimate of the null proportion.

\begin{definition}[\cite{BR07}] \label{def:SC}
	A multiple testing procedure $\bm p \mapsto \cU^*(\bm p)$ is self-consistent with respect to the adaptive threshold collection $\Delta(\bm p, r)$ if 
	\begin{equation}
	p_j \leq \Delta(\bm p, |\cU^*(\bm p)|) \quad \text{for all } j \in \cU^*(\bm p).
	\tag{SC}	
	\label{SC}
	\end{equation}
\end{definition}

%\begin{definition}[\cite{BR08, BR07}] \label{def:SC_old}
%	Let $\widehat m_0(\bm p)$ be an estimate of $m_0 = |\mathcal H_0|$, the number of true null hypotheses. A multiple testing procedure $\cU^*$ is self-consistent with respect to $\widehat m_0$ and $q$ if
%	\begin{equation}
%	p_j \leq \frac{q|\cU^*(\bm p)|}{\widehat m_0(\bm p)} \quad \text{for all } j \in \cU^*(\bm p).
%	\tag{SC}	
%	\label{SC_old}
%	\end{equation}
%\end{definition}

Definitions~\ref{def:thresh_collection} and \ref{def:SC} were proposed by \cite{BR07}, generalizing the corresponding definitions in \cite{BR08} with $G(\bm p) = 1$. While only these original definitions are necessary to prove FDR control for Focused BH, these more general definitions will also allow us to prove FDR control for the adaptive extension presented in Section~\ref{sec:Storey}.

\begin{definition}[\cite{BB14}] \label{def:simple}
	A procedure $\bm p \mapsto \cU^*(\bm p)$ is simple if for any $j \in \cU^*(\bm p)$, $|\cU^*(\bm p)|$ remains unchanged if $p_j$ is varied while $j \in \cU^*(\bm p)$.
\end{definition}

Definition~\ref{def:simple} is a slight variation of that proposed by \cite{BB14}.  We are now ready to state the following lemma, which gives sufficient conditions for a self-consistent multiple testing procedure to control the FDR. Part (iv) of the lemma is novel, while all other parts are known.

\begin{lemma} \label{lem:FDR_control}
	%	Suppose $\bm p$ are valid p-values, and 
	Let $\cU^*$ be a self-consistent multiple testing procedure with respect to an adaptive threshold collection~\eqref{thresh_collection}. $\cU^*$ controls the FDR at level $q$ if any of the following conditions holds, assuming that the p-values $\bm p$ are valid in (i), (iii), (iv), (v):
	\begin{enumerate}
		\item[(i)] $\bm p$ are PRDS, the mapping $\bm p \mapsto |\cU^*(\bm p)|$ is coordinate-wise nonincreasing, $G(\bm p) = 1$, and $\beta(r) = r$;
		\item[(ii)] $\bm p$ are p-values for self-contained hypotheses obtained from valid item-level p-values $\bm p^{\text{item}}$ via the Simes test (recall Section~\ref{sec:setup_hierarchical}), $\bm p^{\text{item}}$ are PRDS, the mapping $\bm p^{\text{item}} \mapsto |\cU^*(\bm p^{\text{item}})|$ is non-decreasing, $G(\bm p) = 1$, and $\beta(r) = r$.
		\item[(iii)] For all $j \in \mathcal H_0$, $p_j$ is independent of $\bm p_{-j},$  the mappings $\bm p \mapsto |\cU^*(\bm p)|$ and $\bm p \mapsto G(\bm p)$ are coordinate-wise nonincreasing, $G$ satisfies \eqref{m_0_hat_criterion}, and $\beta(r) = r$;
		\item[(iv)] For all $j \in \mathcal H_0$, $p_j$ is independent of $\bm p_{-j},$   $\cU^*$ is simple, $G(\bm p) = G(\bm p_{0, j})$ for all $j \in \cU^*(\bm p)$, $G$ satisfies \eqref{m_0_hat_criterion}, and $\beta(r) = r$;
		\item[(v)] $G(\bm p) = 1$, and $\beta(r) = \int_0^r x d\nu(x)$ for some probability measure $\nu$ on $\br^+$.
	\end{enumerate}
\end{lemma}
\begin{proof}
	Part (i) follows directly from Propositions 2.7 and 3.6 of \cite{BR08}, part (ii) follows easily from part (b) of Lemma 2 of \cite{RetJ17} and Proposition 2.7 of \cite{BR08}, part (iii) is the content of Theorem 11 of \cite{BR07}, and part (v) follows directly from Propositions 2.7 and 3.7 of \cite{BR08}. 
	
	It remains to prove part (iv). To this end, first fix $j \in \mathcal H_0$. For fixed $\bm p^0_{-j}$, define the event
	\begin{equation}
	\mathcal E(\bm p^0_{-j}) = \{\bm p: \bm p_{-j} = \bm p^0_{-j}, j \in \cU^*(\bm p)\}.
	\label{simple-event}
	\end{equation}
	By the simpleness criterion, the quantity $|\cU^*(\bm p)|$ is constant for $\bm p \in \mathcal E(\bm p^0_{-j})$. By slight abuse of notation, let $|\cU^*(\bm p^0_{-j})|$ denote this constant value, setting it to zero if $\mathcal E(\bm p_{-j}^0) = \varnothing$. Also note that $G(\bm p) = G(\bm p_{0, j})$ for all $j \in \cU^*(\bm p)$, by assumption. Then, we have
	\begin{equation*}
	\begin{split}
	\EE{\frac{\ind(j \in \cU^*(\bm p))}{|\cU^*(\bm p)|}} &= \EE{\frac{\ind \left(\mathcal E(\bm p_{-j})\right)}{{|\cU^*(\bm p_{-j})|}}\ind\left(p_j \leq \frac{q}{m}|\cU^*(\bm p_{-j})|G(\bm p_{0, j})\right)} \\
	&= \EE{\frac{\ind(\mathcal E(\bm p_{-j}))}{{|\cU^*(\bm p_{-j})|}}\PPst{p_j \leq \frac{q}{m}|\cU^*(\bm p_{-j})|G(\bm p_{0, j})}{\bm p_{-j}}} \\
	&\leq \EE{\frac{\ind(\mathcal E(\bm p_{-j}))}{|\cU^*(\bm p_{-j})|}\frac{q}{m}|\cU^*(\bm p_{-j})|G(\bm p_{0, j})} \\
	&= \frac q m \cdot \EE{\ind(\mathcal E(\bm p_{-j}))G(\bm p_{0,j})} \\
	&\leq \frac q {m_0}.
	\end{split}
	\end{equation*}
	The first equality follows because $\cU^*$ is simple and self-consistent, the second equality by the tower property, the first inequality by the independence assumption and the assumption that p-values are valid, and the second inequality by assumption~\eqref{m_0_hat_criterion}. We conclude that
	\begin{equation*}
	\text{FDR}(\cU^*) = \EE{\frac{\sum_{j \in \mathcal H_0} \ind(j \in \cU^*(\bm p))}{|\cU^*(\bm p)|}} = \sum_{j \in \mathcal H_0} \EE{\frac{\ind(j \in \cU^*(\bm p))}{|\cU^*(\bm p)|}} \leq \sum_{j \in \mathcal H_0} \frac{q}{m_0} = q.
	\end{equation*}
	
\end{proof}
In the following section, we show how Theorems \ref{main_theorem} and \ref{secondary_theorem}, and Corollary~\ref{outer_nodes_corollary}, follow from Lemma~\ref{lem:FDR_control}. 

\subsection{Proofs that Focused BH controls the FDR} \label{sec:supp_FBH_proofs}

\begin{proof}[Proof of Theorem~\ref{main_theorem}]
	
	Let $\bm p \mapsto \cU^*(\bm p)$ represent Focused BH. First, let us establish that Focused BH is a self-consistent procedure with respect to the threshold collection $\Delta(\bm p, r) \equiv \frac{q}{m}r$. Indeed, let $j \in \cU^*$. Then, by the definition of the procedure, 
	\begin{equation}
	\fdphat(t^*) = \frac{m \cdot t^*}{|\cU^*|} \leq q \quad \Longrightarrow \quad p_j \leq t^* \leq \frac{q}{m}|\cU^*|.
	\label{self-consistency-verification}
	\end{equation}
	
	Next, we show that the assumptions in parts (i) or (ii) of Theorem~\ref{main_theorem} are sufficient for the conditions of Lemma~\ref{lem:FDR_control} to hold.
	
	\paragraph{Proof of part (i).} 
	
	By part (i) of Lemma~\ref{lem:FDR_control}, it suffices to show that the mapping $\bm p \mapsto |\cU^*(\bm p)|$ is coordinate-wise nonincreasing. To verify this property, suppose $\bm p^1 \leq \bm p^2$, and let $t_1^*$ and $t_2^*$ be the corresponding p-value cutoffs for Focused BH.
	Let 
	\begin{equation}
	t^{**}_2 = \max\{t \in \{0, p^1_1, \dots, p^1_m\}: t \leq t^*_2\}.
	\label{t2**}
	\end{equation}
	Then, $\cR(t_2^{**}, \bm p^1) = \cR(t_2^{*}, \bm p^1) \supseteq \cR(t_2^*, \bm p^2)$, so
	\begin{equation}
	\fdphat(t^{**}_2, \bm p^1) = \frac{m \cdot t^{**}_2}{|\F(\cR(t^{**}_2, \bm p^1), \bm p^1)|} \leq \frac{m \cdot t^{**}_2}{|\F(\cR(t^*_2, \bm p^2), \bm p^2)|} \leq \frac{m \cdot t^{*}_2}{|\F(\cR(t^*_2, \bm p^2), \bm p^2)|} \leq q,
	\label{monotonicity_argument}
	\end{equation}
	where the first inequality follows by monotonicity of $\F$ and the third by the definition of Focused BH. Therefore, $t^*_1 \geq t^{**}_2$, from which it follows that $\cR(t_1^{*}, \bm p^1) \supseteq \cR(t_2^{**}, \bm p^1) = \cR(t_2^{*}, \bm p^1) \supseteq \cR(t_2^*, \bm p^2)$. By monotonicity it follows that $|\F(\cR(t^*_1, \bm p^1), \bm p^1)| \geq |\F(\cR(t^*_2, \bm p^2), \bm p^2)|$. Therefore, $\bm p \mapsto |\cU^*(\bm p)|$ is indeed non-increasing in each component of $\bm p$.
	
	\paragraph{Proof of part (ii).} By part (ii) of Lemma~\ref{lem:FDR_control}, it suffices to show that the mapping $\bm p^{\text{item}} \mapsto |\cU^*(\bm p^{\text{item}})|$ is coordinatewise non-increasing. This is clear because the mapping $\bm p^{\text{item}} \mapsto \bm p$ is coordinatewise non-decreasing (a property of the Simes test) and the mapping $\bm p \mapsto |\cU^*(\bm p)|$ is coordinatewise non-increasing (shown in the proof of part (i)).
	
	\paragraph{Proof of part (iii).} 
	
	First, suppose $\F$ is monotonic. Then, the statement becomes a special case of part (i). Next, suppose $\F(\cR, \bm p) = \F_0(\cR \cap \mathcal S(\bm p))$. By part (iv) of Lemma~\ref{lem:FDR_control}, it suffices to show that $\cU^*$ is simple. To verify this property, fix $j$ and let $\bm p^1$ and $\bm p^2$ be p-value vectors differing only in their $j$th coordinate such that $j \in \cU^*(\bm p^1) \cap \cU^*(\bm p^2)$. Suppose without loss of generality that $t^*_1 \leq t^*_2$. Also, define $t_2^{**}$ via (\ref{t2**}) as before. Then, note that $\cR(t_2^{**}, \bm p^1) = \cR(t_2^{*}, \bm p^1) = \cR(t_2^*, \bm p^2)$. Also, $\mathcal S(\bm p^1) = \mathcal S(\bm p^2)$ because $\mathcal S$ is stable and $j \in \cU^*(\bm p^1) \cap \cU^*(\bm p^2) \subseteq \mathcal S(\bm p^1) \cap \mathcal S(\bm p^2)$. Therefore, 
	\begin{equation*}
	\begin{split}
	\fdphat(t_2^{**}, \bm p^1) = \frac{m \cdot t_2^{**}}{|\F(\cR(t_2^{**}, \bm p^1), \bm p^1)|} &= \frac{m \cdot t_2^{**}}{|\F_0(\cR(t_2^{**}, \bm p^1) \cap \mathcal S(\bm p^1))|} \\
	&= \frac{m \cdot t_2^{**}}{|\F_0(\cR(t_2^*, \bm p^2) \cap \mathcal S(\bm p^2))|} \\
	&= \frac{m \cdot t_2^{**}}{|\F(\cR(t_2^*, \bm p^2), \bm p^2)|} \leq \frac{m \cdot t_2^{*}}{|\F(\cR(t_2^*, \bm p^2), \bm p^2)|} = \fdphat(t_2^*, \bm p^2) \leq q,
	\end{split}
	\end{equation*} 
	where the first inequality holds because $t_2^{**} \leq t_2^*$ by construction. From this it follows that $t_2^{**} \leq t_1^* \leq t_2^*$, so $\cR(t_1^*, \bm p^1) = \cR(t_2^*, \bm p^1) = \cR(t_2^*, \bm p^2)$. Therefore, we find that $\F(\cR(t^*_1, \bm p^1), \bm p^1) = \F_0(\cR(t^*_1, \bm p^1) \cap \mathcal S(\bm p^1)) = \F_0(\cR(t^*_2, \bm p^2) \cap \mathcal S(\bm p^2)) = \F(\cR(t^*_2, \bm p^2), \bm p^2)$. Focused BH is thus a simple multiple testing procedure, which completes the proof.
\end{proof}

\begin{proof}[Proof of Corollary~\ref{outer_nodes_corollary}]
	By part (i) of Theorem \ref{main_theorem}, it suffices to check that $\F_{\mathcal T}$ is monotonic. To this end, suppose $\cR^1 \supset \cR^2$. We must show that $|\F_{\cT}(\cR^1)| \geq |\F_{\cT}(\cR^2)| $. Suppose first that $\cR^1$ is obtained from $\cR^2$ by adding one node $j$. If $j$ is not an outer node of $\cR^1$, then each of the outer nodes of $\cR^2$ are still outer nodes of $\cR^1$, in which case $|\F_{\mathcal T}(\cR^1)| = |\F_{\mathcal T}(\cR^2)|$. If $j$ is an outer node of $\cR^1$, then it can have at most 1 ancestor that is an outer node of $\cR^2$, since the graph is a tree. Hence, in this case either $|\F_{\mathcal T}(\cR^1)| = |\F_{\mathcal T}(\cR^2)|$ or $|\F_{\mathcal T}(\cR^1)| = |\F_{\mathcal T}(\cR^2)| + 1$. Having addressed the case $|\cR^1| = |\cR^2| + 1$, the general case follows by induction.
\end{proof}

\begin{proof}[Proof of Theorem~\ref{secondary_theorem}]
	As with the proof of Theorem \ref{main_theorem}, we rely on Lemma~\ref{lem:FDR_control}. The self-consistency of the Focused BH base procedure, which outputs $\cR(t^*, \bm p)$, with respect to $\Delta(\bm p, r) = \frac{q}{m}r$, is easily verified along the same lines as \eqref{self-consistency-verification}. Under the assumptions of part (i) or part (ii), monotonicity follows because as we showed before, $\bm p^1 \leq \bm p^2$ implies that $\mathcal R(t_1^*, \bm p^1) \supseteq \mathcal R(t_2^*, \bm p^2)$. Under the assumptions of part (iii), we claim that the Focused BH base procedure is simple. Indeed, fix $j$ and let $\bm p^1$ and $\bm p^2$ be p-value vectors differing only in their $j$th coordinate such that $j \in \cR(\bm p^1, t_1^*) \cap \cR(\bm p^2, t_2^*)$. Defining $t_2^{**}$ as before, we again have that $\cR(t_2^{**}, \bm p^1) = \cR(t_2^*, \bm p^2)$. By our assumption on $\F$, it follows that $|\F(\cR(t_2^{**}, \bm p^1), \bm p^1)| = |\F(\cR(t_2^*, \bm p^2), \bm p^2)|$, from which it follows as before that $\cR(t_1^*, \bm p^1) = \cR(t_2^*, \bm p^2)$. Therefore, the Focused BH base procedure is a simple multiple testing procedure, which completes the proof.
\end{proof}

\section{Discussion of resampling approach} \label{sec:supp_perm}

In this section, we describe a set of conditions under which the approximation~\eqref{permutation_derivation} holds, and therefore the permutation version of Focused BH is likely to succeed. 

The first key condition is that the distribution of the null p-values remains unchanged after permutation; i.e.
\begin{equation}
\{p_j\}_{j \in \mathcal H_0} \overset d = \{\tilde p_j\}_{j \in \mathcal H_0}, \label{A1}
\end{equation}
where $\{\tilde p_j\}$ are p-values derived from the permuted data set. This is Westfall and Young's subset pivotality property \cite{WY93}. It allows us to treat $\{\tilde p_j\}_{j \in \mathcal H_0}$ like samples from the true null distribution. See \cite{RY13} for a discussion of when subset pivotality holds in gene expression experiments of the kind outlined in Section~\ref{sec:improving} of the main text.

Secondly, suppose the inclusion of a null hypothesis in the filtered set at any fixed p-value threshold is a function only of $\{p_{j'}\}_{j' \in \cH_0}$:
\begin{equation}
\ind(j \in \cU(t, \bm p)) \text{ depends only on } \{p_{j'}\}_{j' \in \mathcal H_0} \text{ for each } j \in \cH_0. 
\label{A2}
\end{equation}
Taken together, assumptions \eqref{A1} and \eqref{A2} imply that
\begin{equation*}
\mathbb E\left[|\cU(t, \bm p) \cap \mathcal H_0|\right] = \EE{\sum_{j \in \mathcal H_0}\ind(j \in \cU(t, \bm p))} = \EE{\sum_{j \in \mathcal H_0}\ind(j \in \cU(t, \tilde{\bm p}))} = \mathbb E\left[|\cU(t, \tilde{\bm p}) \cap \mathcal H_0|\right],
\end{equation*}
which is the key step in the derivation~\eqref{permutation_derivation}. Therefore, under these two assumptions, we can conclude that $\widehat V^{\text{perm}}(t)$ is a conservative estimate of $V(t)$, i.e.
\begin{equation*}
\mathbb E[V^{\text{perm}}(t)] \geq \mathbb E[V(t)].
\end{equation*}
This fact by itself does not imply FDR control for the permutation procedure, but at least suggests FDR control might hold. 

Note that (\ref{A1}) is an assumption on the data-generating process and the permutation mechanism, and has been extensively studied. On the other hand, (\ref{A2}) is an assumption on the filter. For example, this assumption holds for the outer nodes filter if logical relationships (\ref{logical_relationships}) between hypotheses and their descendants hold. Indeed, these logical relationships in the DAG ensure that all null nodes are below all non-null nodes, so intuitively the outer nodes filter considers null nodes before non-null nodes. Formally, to prove (\ref{A2}) it suffices to observe that $j \in \mathcal U(t, \bm p)$ if and only if $p_j \leq t$ and $p_{j'} > t$ for each $j \leadsto j'$. If $j$ is null, then all its descendants $j'$ are null due to (\ref{logical_relationships}), and so the quantity $\ind(j \in \cU(t, \bm p))$ depends only on null p-values.

We leave the investigation of the theoretical properties of the permutation approach for future work.

\section{Other extensions of Focused BH} \label{sec:extensions}

\subsection{Focused Storey-BH} \label{sec:Storey}

Like BH, Focused BH may be improved by adaptively estimating and correcting for the proportion of non-nulls. Following \cite{Storey02, Storey04} a family of estimators of $m_0$ can be defined as:
\begin{equation*}
\widehat m_0^\lambda \equiv \widehat m_0^\lambda(\bm p) \equiv \frac{1 + |\{j: p_j > \lambda\}|}{1-\lambda}; \quad \lambda \in (0,1).
\end{equation*}
For a given $\lambda \in (0,1)$, we may replace (\ref{FDP_hat}) with
\begin{equation}
\fdphat^{\text{Storey}}(t) \equiv \frac{\widehat m_0^{\lambda} \cdot t}{|\F(\cR(t, \bm p), \bm p)|}.
\label{Storey-est}
\end{equation}
Like the Storey-BH method \cite{Storey04}, we require that $t^* \leq \lambda$. Therefore, we define this threshold as
\begin{equation}
t^* \equiv \max\{t \in \{0, p_1, \dots, p_m\} \cap [0,\lambda]: \fdphat^{\text{Storey}}(t) \leq q\}.
\label{t_star_Storey}
\end{equation}
%We define the Focused Storey-BH procedure via \eqref{Storey-est} and \eqref{t_star_Storey}.
The Focused Storey-BH procedure is defined by  replacing lines 2 and 4 in the Focused BH algorithm by \eqref{Storey-est} and \eqref{t_star_Storey}, respectively. FDR control for the usual Storey-BH procedure (without filtering) is proven only under independence \cite{Storey04}. We can show that under this assumption and for a broad class of filters, Focused Storey-BH also controls the FDR:
\begin{theorem} \label{storey-theorem}
	Focused Storey-BH controls the FDR under the assumptions of Theorem~\ref{main_theorem}, part (iii): $p_j \independent \bm p_{-j}$ for all $j \in \mathcal H_0$ and the filter is monotonic or of the form $\F(\cR, \bm p) = \F_0(\cR \cap \mathcal S(\bm p))$, where $\F_0$ is an arbitrary fixed filter and $\mathcal S$ is a \textit{stable} screening function.
\end{theorem}
\begin{proof}
	First, it is easy to verify that Focused Storey-BH is self-consistent with respect to the adaptive threshold collection 
	\begin{equation*}
	\Delta(\bm p, r) = \frac{q}{m}\cdot r\cdot\frac{m}{\widehat m_0^\lambda(\bm p)}.
	\end{equation*}
	
	We apply parts (iii) and (iv) of Lemma~\ref{lem:FDR_control} to prove the theorem. Note that criterion~\eqref{m_0_hat_criterion} was proved for $G(\bm p) = m/\widehat m_0^{\lambda}(\bm p)$ by \cite{YetY06}.
	
	Suppose $\F$ is monotonic. Part (iii) of Lemma~\ref{lem:FDR_control} shows that it is sufficient to show that Focused Storey-BH and the estimate $\widehat m_0^{\lambda}$ are monotonic. The proof of the latter is trivial, and the proof of the former follows from a similar argument to that of Theorem~\ref{main_theorem}, part (i).
	
	Suppose $\F(\cR, \bm p) = \F_0(\cR \cap \mathcal S(\bm p))$. Part (iv) of Lemma~\ref{lem:FDR_control} shows that it is sufficient to show that Focused Storey-BH is simple and that $\widehat m_0(\bm p) = \widehat m_0(\bm p_{0,j})$ for all $j \in \cU^*(\bm p)$. The latter statement is clear because $j \in \cU^*(\bm p)$ implies that $p_j \leq \lambda$, so setting $p_j$ to zero will not impact $\widehat m_0^\lambda$. Given this fact, a minor modification of the proof of Theorem~\ref{main_theorem}, part (iii) shows that $\cU^*$ is simple.	
\end{proof}

We remark that other estimates $\widehat m_0$ are possible; see \cite{BR07}. We leave the theoretical investigation of using these estimates with Focused BH for future work. 

\subsection{Focused Reshaped BH} \label{sec:reshaping}

BH is known to control the FDR under independence and positive dependence, but for a theoretical guarantee of FDR control under arbitrary dependence, an extra correction is required \cite{BY01}. This kind of correction can be formulated generally via the shape function \cite{BR08}
\begin{equation}
\beta(r) = \int_0^r x d\nu(x),
\label{reshaping}
\end{equation}
as described in part (v) of Lemma~\ref{lem:FDR_control}. Since $\beta(r) \leq r$, the idea is to use $\beta$ to \textit{undercount} the number of discoveries in the denominator of $\fdphat$ in order to make the procedure more conservative. This protects against adversarial dependency structures. For example, the measure $\nu$ placing masses in proportion to $\frac{1}{j}$ on each $j$ leads to $\beta(j) = \frac{j}{\sum_{i = 1}^m\frac{1}{i}}$. This shape function reduces to the Benjamini and Yekutieli's correction \cite{BY01}. Many other shape functions are possible; see \cite{BR08} for a discussion. We can extend these ideas to Focused BH as well: For any shape function $\beta$ of the form \eqref{reshaping}, Focused Reshaped BH is defined via
\begin{equation*}
\fdphat^{\beta}(t) \equiv \frac{m \cdot t}{\beta(|\F(\cR(t, \bm p), \bm p)|)},
\end{equation*}
keeping the other components of Focused BH the same. 
\begin{theorem} \label{reshaping-theorem}
	Focused Reshaped BH controls the FDR for arbitrary filters and valid p-values with arbitrary dependence structure.
\end{theorem}
\begin{proof}
	This statement follows from part (v) of Lemma~\ref{lem:FDR_control} once we note that Focused Reshaped BH is self-consistent with respect to the threshold collection ${\Delta(\bm p, r) = \frac{q}{m}\cdot \beta(r)}$.
\end{proof}
As with other procedures involving reshaping, Focused Reshaped BH might be very conservative. If p-value dependency is a major concern, a simpler solution might be to apply an FWER correction instead, after which arbitrary filtering may be done. Note that Focused Reshaped BH might or might not be more powerful than the FWER approach; e.g. see \cite{BY01} for a discussion of the power of their reshaped BH method.

\subsection{Multiple filters}

We may have multiple filters of interest, $\F_1, \dots, \F_L$. In this case, we may want a base procedure $\cM_0$ such that each composite procedure $\cM_\ell = \F_\ell \circ \cM_0$ has FDR control at pre-specified target level $q_\ell$, for $\ell = 1, \dots, L$. If $\mathcal M_0$ rejects the set $\cR(t^*, \bm p)$ for a threshold $t^*$, this amounts to
\begin{equation}
\text{FDR}_{\ell} \equiv \mathbb E[\text{FDP}(\F_\ell(\cR(t^*, \bm p), \bm p))] \leq q_\ell \quad \text{for each } \ell = 1, \dots, L,
\label{multi-filter-FDR}
\end{equation}
This criterion is similar to the multilayer FDR control criterion introduced by \cite{FBR15}. Note that Theorem \ref{secondary_theorem} (in conjunction with Theorem \ref{main_theorem}) implies that the Focused BH base procedure controls FDR with respect to the trivial filter ($\F_1$) and any monotonic filter $(\F_2$) under PRDS p-values. For a general set of filters, we can define the \textit{Multi-Focus BH} procedure by analogy to the p-filter \cite{FBR15} and multilayer knockoff filter \cite{katsevich17mkf}. Define
\begin{equation*}
t^* \equiv \max\{t \in \{0, p_1, \dots, p_m\}: \widehat{\text{FDP}}_\ell(t) \leq q_\ell \text{ for all } \ell\},
\end{equation*}
where $\widehat{\text{FDP}}_\ell(t)$ are defined as in (\ref{FDP_hat}), one for each filter. The Multi-Focus BH procedure outputs $L$ rejection sets, where each one is obtained by applying one of the pre-specified filters on the base rejection set $\cR(t^*, \bm p)$.
\begin{theorem} \label{multi-focus-bh}
	If each filter $\F_\ell$ is monotonic in $\bm p$ and the p-values are PRDS, then Multi-Focus BH satisfies \eqref{multi-filter-FDR}.
\end{theorem}
\begin{proof}
	The proof of this theorem is similar to that of Theorem \ref{main_theorem} part (i). By the same argument, it suffices to show that $\bm p^1 \leq \bm p^2$ implies that $\mathcal R(t_1^*, \bm p^1) \supseteq \mathcal R(t_2^*, \bm p^2),$
	%$t^*$ is monotonic in $\bm p$, 
	and this fact is derived analogously to (\ref{monotonicity_argument}).
\end{proof}
An appealing approach may be to consider all the outputs of Multi-Focus BH, and report only one output which gives the most interesting results, e.g. the one corresponding to the largest, or the most easily interpretable, rejection set. This approach is not valid in the sense that we have no FDR control guarantees for a filter which is chosen among the $L$ filters based on the data. In order to allow a post-hoc choice of the filter, one could apply Multi-Focus BH at levels $q_1 = \cdots = q_L \equiv q/L$. Indeed, under the assumptions of   Theorem~\ref{multi-focus-bh}, we would have
\begin{equation*}
\EE{\max_{\ell=1, \ldots, L}\fdp (\F_{\ell}(\mathcal{R}(t^*, \bm p), \bm p))}\leq \sum_{\ell=1}^L \fdr_\ell \leq q;
\end{equation*}
i.e. Multi-Focus BH at level $q/L$ guarantees FDR control for any of its output rejection sets, even if it is chosen based on the data. A natural competitor of such procedure is a FWER-controlling procedure, which
allows any filtering of its rejection set, preserving FWER (and FDR) guarantees. If a practitioner is interested in FDR control, and knows in advance that only a small number of filters may be of interest, it may be preferable to apply Multi-Focus BH rather than a FWER-controlling procedure, especially if there are many signals.

\section{Hierarchical simulations and data analysis}

In this section, we present supplementary information on the numerical simulations and data analysis carried out in Sections~\ref{sec:experiments} and \ref{sec:data_analysis} of the main text, respectively.

\subsection{Supplementary figures}

Figure~\ref{fig:permutation_estimates} shows the estimates $\widehat V$ given in (\ref{BH_derivation}) in the main text for the GO and ICD graphs (normalized by their respective numbers of nodes for direct comparison) as well as the original linear estimate. Figure~\ref{fig:BH_experiments} compares the performance of Focused BH to that of methods not accounting for the hierarchical structure of the multiple testing problem in the ``intermediate" simulation setting of Section~\ref{sec:ICD_experiment} in the main text. Figures~\ref{fig:ICD_analysis_graph_info} and \ref{fig:REVIGO_analysis_graph_info} show information about the graphs used in the PheWAS and GO data analyses, respectively, of Section~\ref{sec:data_analysis}.

\begin{figure}[h!]
	\centering
	\includegraphics[width = 0.7\textwidth]{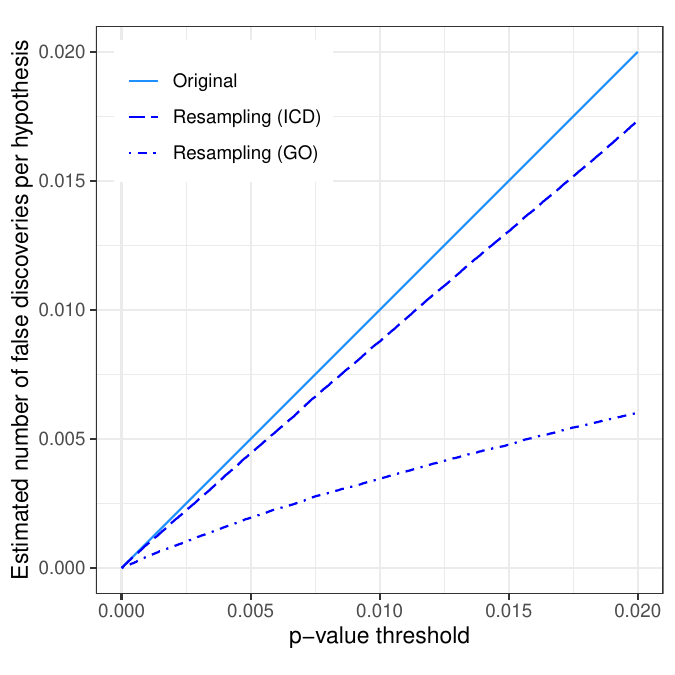}
	\caption{Resampling estimates compared to BH estimate for the two simulation scenarios.}
	\label{fig:permutation_estimates}
\end{figure}
\begin{figure}[h!]
	\centering
	\includegraphics[width = 0.9\textwidth]{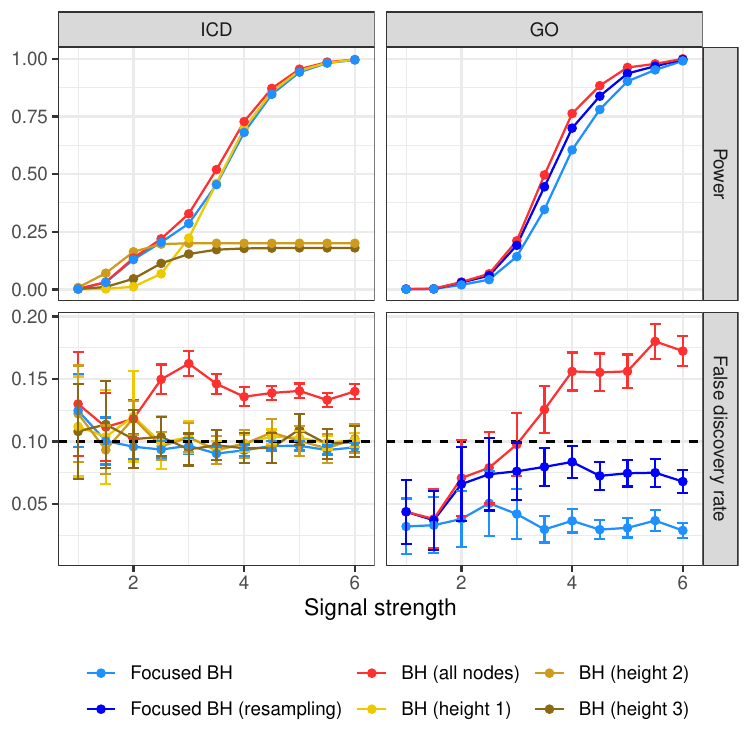}
	\caption{FDR and power of applying BH to either all nodes or to nodes at each depth, followed by filtering, compared to that of Focused BH, in the intermediate setting of the PheWAS simulation (left panels) or the GO simulation (right panels) from Section~\ref{sec:experiments} of the main text.}
	\label{fig:BH_experiments}
\end{figure}

\begin{figure}[h!]
	\centering
	\includegraphics[width = \textwidth]{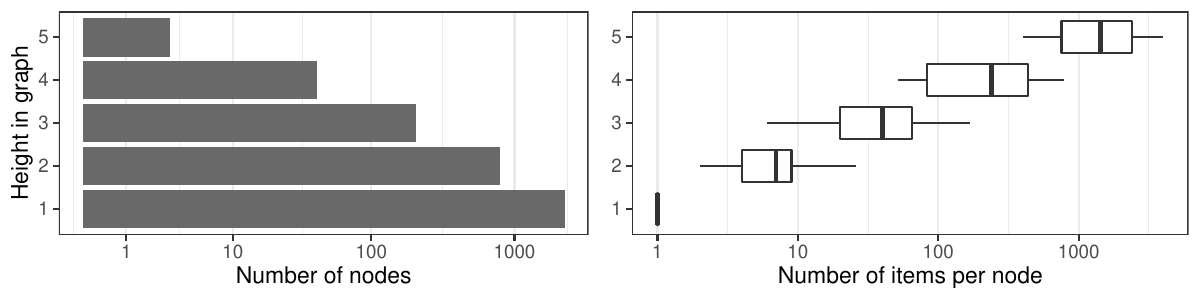}
	\caption{Summary of the ICD subgraph used for data analysis. Number of nodes (left) and numbers of items per node (right) at each level of the tree.}
	\label{fig:ICD_analysis_graph_info}
\end{figure}
\begin{figure}[h!]
	\centering
	\includegraphics[width = \textwidth]{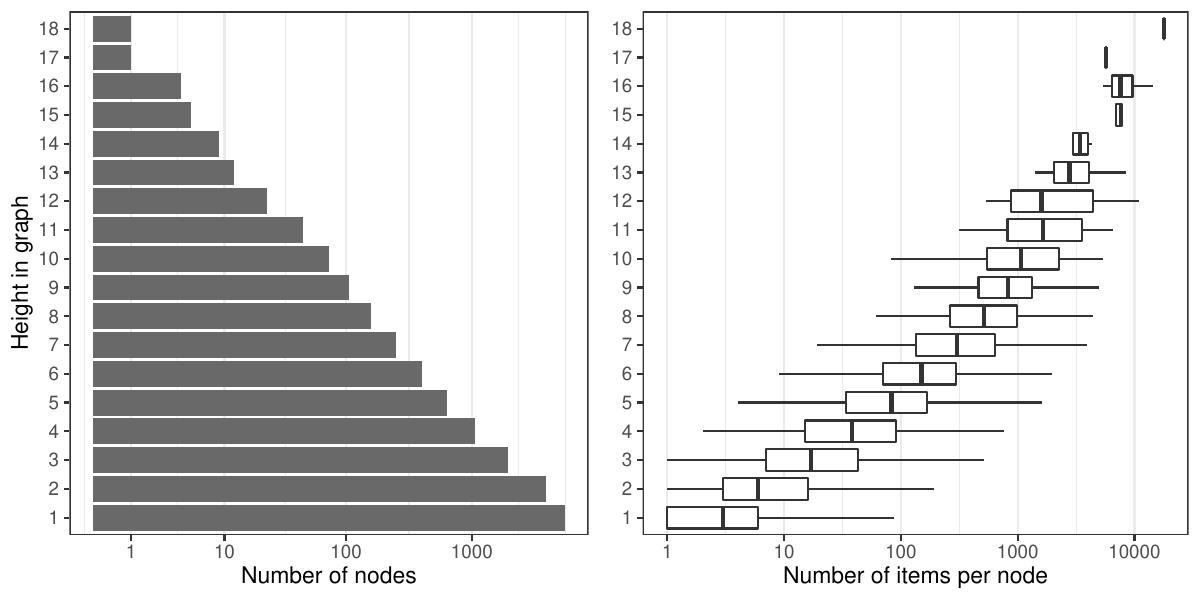}
	\caption{Summary of the GO subgraph used for data analysis. Number of nodes (left) and numbers of items per node (right) at each level of the tree.}
	\label{fig:REVIGO_analysis_graph_info}
\end{figure}

\clearpage

\subsection{Examining novel PheWAS discoveries} \label{sec:UKBB_discoveries}

In the PheWAS data analysis we carried out in Section~\ref{sec:phewas} of the main text, there were four discoveries made by Focused BH but not Leaf BH, Structured Holm, or Yekutieli:
\begin{enumerate}
	\item ``Other synovitis and tenosynovitis";
	\item  ``Chronic sinusitis";
	\item ``Other benign neoplasms of connective and other soft tissue";
	\item ``Symptoms, signs and abnormal clinical and laboratory findings, \\ not elsewhere classified."
\end{enumerate}
In this section, we examine the extent to which existing literature supports these findings.

The first of these is plausible since reactive arthritis, known to be associated with HLA-B*27, is hypothesized to be a form of synovitis \cite{Colmegna2004}. The second disease, chronic sinusitis, has been found to be comorbid with HLA-B*27 associated traits like reactive arthritis \cite{Ngaruiya2013}, and therefore may represent a true, but indirect, association. The association of the third disease with the variant in question is somewhat more tenuous, though other HLA variants have been found to be associated with a neoplasm called meningioma \cite{Machulla2003}. Finally, the last of these associations is reasonable because HLA-B*27 may have clinical manifestations that cannot be assigned to a well-defined diagnosis.

We remark that much more work needs to be done to rigorously validate these findings, but it is beyond the scope of this paper.

\section{GWAS, an application with spatial structure} \label{sec:GWAS_simulation}

To supplement the hierarchical examples in the main text, we consider spatially structured GWAS data.

\subsection{Structure and filtering in GWAS}

The goal of a GWAS is to identify genomic loci where genetic variation is associated with variation in phenotypes: this is carried out by testing the possible association between each of hundreds of thousands of genetic markers (typically single nucleotide polymorphisms, or SNPs) and a trait of interest. Due to spatially localized correlation patterns in the genome called linkage disequilibrium (LD), significant SNPs often come in ``clumps." While none of these SNPs (which are genotyped for their high variability in the population and the ease with which they are assayed) are necessarily mechanistically related with the phenotype of interest, nearby correlated SNPs act as proxies for an unmeasured ``causal'' variant. Therefore, significant SNPs are grouped into \textit{loci}, i.e. groups of nearby SNPs, to be explored in follow-up studies. Each locus is usually represented by the SNP with the most significant p-value, called the ``lead SNP." We refer to this as the ``clumping filter." Software tools like SWISS \cite{SWISS} have been developed with this goal.

For illustration purposes, consider a slight simplification of clumping filters used in practice. The correlation pattern of SNPs is roughly block-diagonal, say with blocks $\mathcal B_1, \dots, \mathcal B_T$ partitioning the genome. We assume that these blocks are defined prior to testing (for example relying on information on linkage disequilibrium available from external sources, or clustering SNPs in the dataset without regard of the outcome value). The {\em GWAS clumping filter} is defined by
\begin{equation*}
\F(\cR, \bm p) = \{j \in \cR: p_j \leq p_{j'} \text{ for all } j' \in \mathcal B_{t(j)}\},
\end{equation*}
where $t(j) \in [T]$ is the block containing SNP $j$. Therefore, this filter keeps the top SNP in each block. Note that for sets $\cR = \cR(t, \bm p)$, the modified clumping filter is equivalent to the screening filter with 
\begin{equation}
\mathcal S(\bm p) = \{j \in [m]: p_j \leq p_{j'} \text{ for all } j' \in \mathcal B_{t(j)}\}.
\label{screening_GWAS}
\end{equation}
It follows that Focused BH with the modified clumping filter is equivalent to Focused BH with the screening filter based on (\ref{screening_GWAS}).

The GWAS clumping filter is monotonic. To see this, note that
\begin{equation}
|\F(\cR, \bm p)| = \sum_{t = 1}^T \ind(\mathcal B_t \cap \cR \neq \varnothing),
\label{clumping_norm}
\end{equation}
i.e. the number of lead SNPs is equal to the number of blocks intersecting the candidate rejection set $\cR$. From this representation, monotonicity is easy to verify.  Therefore, assuming positive dependence of the p-values, Theorem \ref{main_theorem} part (i) assures us that Focused BH controls FDR in this context. As discussed by \cite{BetS17a} (whose result we generalize) the PRDS assumption is difficult to check in practice but it can be expected to hold at least approximately given the positive correlation structure in the genome.

\subsection{Focused BH with simulated GWAS data}

Next, we apply Focused BH with the GWAS clumping filter in a stylized simulation setting. 

\paragraph{Data generating mechanism.} We generate synthetic genotype  data $X \in \{0, 1, 2\}^{n \times m}$  for  $n = 500$ individuals and $m = 3000$ SNPs, following a haplotype block model.
The individual genotypes (rows of $X$) are drawn i.i.d. from a genotype distribution, modeled as the sum of two independent haplotypes (i.e. vectors in $\{0, 1\}^m$). Each haplotype, in turn, is drawn from a Markov chain on $\{0, 1\}$. This Markov chain is constructed so that SNPs come in correlated blocks of size 30, these blocks being independent of each other. The distribution of each block is a stationary Markov chain, with marginal distributions $\PP{X_{ij} = 1} = 0.1$ (i.e. the minor allele frequency of each SNP is 0.1) and transition matrix such that $\PP{X_{i,j+1} = 1|X_{i,j} = 1} = 0.95$. This choice models strong linkage disequilibrium within each block.

Once the genotype matrix is created, the phenotype is generated using  the linear model
\begin{equation*}
y = X\beta + \eps; \quad \eps \sim N(0, I).
\end{equation*}
For each parameter setting, we generate the genotype matrix $X$ once and repeatedly generate $y$ from the above distribution of $y|X$. The values of  $\beta_j$ are  chosen to represent  a set of 10 ``causal SNPs" $\mathcal S \subseteq [m]$ spaced equally along the ``genome:"
\begin{equation*}
\beta_j = \begin{cases}
A \quad \text{if } j \in \mathcal S; \\
0 \quad \text{otherwise.}
\end{cases}
\end{equation*}

Given a genotype matrix $X$ and a phenotype matrix $y$, we obtain p-values for each SNP using the  marginal testing approach that is standard in GWAS. That is, for each $j$ we consider the (misspecified) linear model
\begin{equation*}
y = b_0 + b_j X_j + \eta, \quad \eta \sim N(0, \sigma^2)
\end{equation*}
and use the usual two-sided t-test for the hypothesis $H_0: b_j = 0$. Here, $X_j$ is the column of $X$ corresponding to SNP $j$. 
The null hypothesis tested by this approach is that of no association between SNP $j$ and the phenotype. Taking into account the linkage disequilibrium between SNPs and the phenotype generating process,
the collection of non-null SNPs $\mathcal H_1$ contains all the SNPs $j$ that are  in the same LD group (block) as a causal SNP. 

Note that since the non-null identities of all members in each LD block are the same, it really  doesn't matter which SNP in a block the clumping filter chooses; we can also think of the filter as taking a set of significant SNPs and returning a set of significant ``loci."

\paragraph{Methods compared.} Since the graph-based methods (Yekutieli and Structured Holm) are not applicable to spatially-structured hypotheses, we apply only BH followed by filtering, and three variants of Focused BH: regular, permutation-based, and oracle. The permutation-based method is as described in Section~\ref{sec:improving} of the main text, with genotypes playing the role of gene expression in the example given there. The oracle version is based on the estimate~\eqref{V_hat_oracle}, which we use a theoretical benchmark.

\paragraph{Results.}

Figure \ref{fig:GWAS_results} shows the power and FDR of the methods considered. It is clear that BH followed by filtering loses FDR control quite dramatically. All of the variants of Focused BH control the FDR in this case. This includes the permutation-based version of Focused BH, echoing the same conclusion from Figure~\ref{fig:REVIGO_experiment} and reinforcing our conjecture that this methodology does control the FDR. Excluding BH due to its lack of calibration, the oracle and permutation versions of Focused BH have the highest power, followed by the original version. It is remarkable that the permutation-based version of Focused BH attains near-oracle performance in this simulation.

Figure \ref{fig:Manhattan} and Table~\ref{table:Manhattan} show the results of one run of this simulation (with $A = 0.45$). It is clear how  linkage disequilibrium translates into multiple adjacent SNPs to have  significant p-values. The BH method, ignoring this, sets the p-value threshold too optimistically, incurring many false discoveries. Focused BH corrects for the filter, but in this case is somewhat conservative. The permutation procedure provides the best trade-off: making few false discoveries while achieving higher power.

\begin{figure}[h!]
	\centering
	\includegraphics[width = 0.8\textwidth]{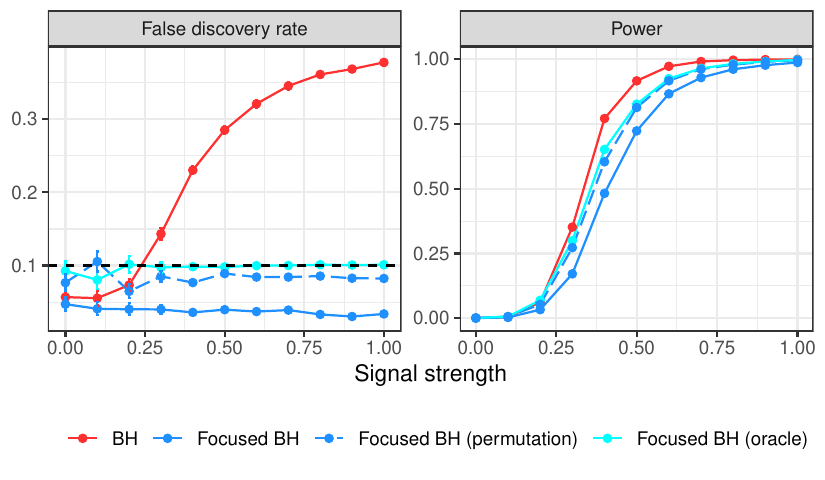}
	\caption{FDR and power for the methods compared on the
		GWAS simulation using the clumping filter. The target FDR for each procedure is 0.1; reported are the averages of 500 replications.}
	\label{fig:GWAS_results}
\end{figure}

\begin{figure}[h!]
	\centering
	\includegraphics[width = 0.8\textwidth]{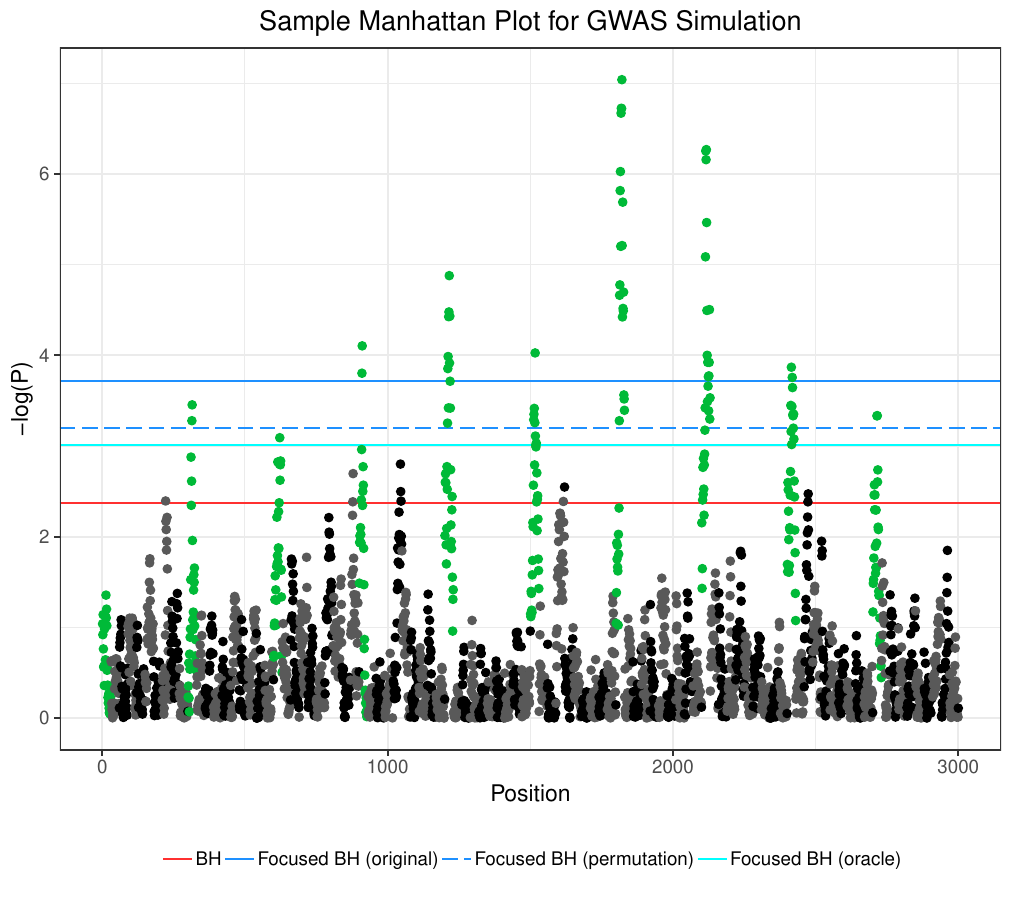}
	\caption{Manhattan plot illustrating one replicate of the GWAS-clumping simulation with $A = 0.45$. Log10 of the p-values for non-null SNPs are shown in green; and for null SNPs in black and gray, with alternating colors across different  LD blocks. The horizontal lines represent the p-value cutoff of the different methods under consideration, color coded as in Figure \ref{fig:GWAS_results}.}
	\label{fig:Manhattan}
\end{figure}

\begin{table}[h!]
	\centering
	\begin{tabular}{l|ccc}
		Method & FDP before filtering & $\fdp$ after filtering & Power after filtering\\
		\hline
		\textcolor{firebrick1}{BH} & 0.08 & 0.36 & 0.9 \\
		\textcolor{dodgerblue}{Focused BH (original)} & 0 & 0 & 0.6\\
		\textcolor{dodgerblue}{Focused BH (permutation)} & 0 & 0 & 0.8\\
		\textcolor{cyan}{Focused BH (oracle)} & 0 & 0 & 0.9 \\
	\end{tabular}
	\caption{FDP (pre- and post-filtering) as well as the (realized) post-filtering power of BH and three versions of Focused BH on the problem instance in Figure~\ref{fig:Manhattan}.}
	\label{table:Manhattan}
\end{table}

\section{From subsetting to prioritization} \label{sec:app_fractional}

As discussed in Section~\ref{sec:discussion} of the main text, Focused BH may be extended to filters that prioritize rejections in addition to subsetting them. We define this more general kind of filter in Section~\ref{sec:supp_extending} and then give an example in Section~\ref{sec:supp_soft_outer_nodes}.

\subsection{A more general definition of filtering} \label{sec:supp_extending}

\begin{definition}
	Given a set of p-values $\bm p$ and a subset $\cR \subseteq \mathcal H,$ a \textit{generalized filter} $\F$ is a map ${\F}: (\cR, \bm p)\mapsto \bm U \in [0,1]^m$  with the property that
	\begin{equation*}
	U_j = 0 \quad \text{if} \quad j \not \in \cR.
	\end{equation*}
\end{definition}

The vector $\bm U \in [0,1]^m$ is a \textit{prioritization vector}, attaching an importance to each hypothesis in $\mathcal H$ (the larger $U_j$, the more important $H_j$ is deemed). The requirement that $U_j = 0$ for $j \not \in \cR$ reflects the fact that the prioritization really occurs on the set $\cR$ and the rest of the hypotheses get a weight of zero. 
%Note that a filter for which $U_j \in \{0,1\}$ can be identified with a subsetting filter via $\mathcal U = \{j: U_j = 1\}$.
While for a subsetting filter $U_j \in \{0,1\}$,  continuous-valued $U_j \in (0, 1)$ can describe filters which do not necessarily remove rejections altogether but instead assign to each some measure of importance. For example, given a set of fixed weights $u_1, \dots, u_m \in [0,1]$, the \textit{fixed weights filter} is defined via
\begin{equation}
U_j \equiv u_j \ind(j \in \cR).
\label{fixed_weights}
\end{equation}
This is the simplest example of a filter for which $U_j \not \in \{0,1\}$, and in this case $\bm U$ carries the same amount of information as the set $\mathcal U = \{j: U_j > 0\}$. We present a more nontrivial example of a prioritization filter in the next section. 

The composition of a base multiple testing procedure with a generalized filter inputs a set of p-values and outputs a prioritization. This can be viewed as a generalized multiple testing procedure, whose FDP for any output $\bm U$ can naturally be defined via
\begin{equation}
\textnormal{FDP}(\bm U) \equiv \frac{\sum_{j = 1}^m U_j \ind(j \in \mathcal H_0)}{\sum_{j = 1}^m U_j},
\label{weighted_FDP}
\end{equation}
where we define $0/0 \equiv 0$. The sum in the denominator $\sum_{j = 1}^m U_j$ represents the total number of discoveries slated for further investigation, or total amount of resources devoted to follow-up, or the total number of distinct scientific findings that will be discussed in the paper summarizing results. The FDP is then the fraction of this total devoted to null hypotheses. Note that the definition \eqref{weighted_FDP} is the same as that of the hypothesis-weighted FDP in previous works, such as \cite{BH97, BR08}, when the prioritization vector is given by \eqref{fixed_weights}. 

Now, may we extend Focused BH to generalized filters, with the only change that 
\begin{equation*}
\fdphat(t) \equiv \frac{m \cdot t}{\norm{\F(\{j: p_j \leq t\}, \bm p)}}, \quad \text{where } \norm{\bm U} \equiv \sum_{j = 1}^{m} U_j.
\end{equation*}
These fractional weights pose minimal extra theoretical challenges, and in fact the following theorem (whose proof we omit for the sake of brevity) states that all our previous results continue to hold for this more general case. 

\begin{theorem} \label{fractional_theorem}
	The results of Theorems \ref{main_theorem}, \ref{secondary_theorem}, \ref{storey-theorem}, \ref{reshaping-theorem}, and \ref{multi-focus-bh}---if $|\F(\cR, \bm p)|$ is replaced with $\norm{\F(\cR, \bm p)}$ where appropriate---continue to hold for generalized filters.
\end{theorem}

\subsection{Example: The soft outer nodes filter} \label{sec:supp_soft_outer_nodes}

As we discussed in Section~\ref{sec:theoretical_results} of the main text, the outer nodes filter is not monotonic on general DAGs. This suggests that more information is contained in the rejection of two parent nodes than in the rejection just of their child, so other filters may be more appropriate in this context. In this section, we propose the \textit{soft outer nodes filter}, a prioritization filter that is monotonic on general DAGs. 

The starting point for the soft outer nodes filter is the remark made by \cite{PetC09}: to appropriately interpret the information in the GO, one cannot simply look at the graph structure (``edge information''), but needs to take explicitly into consideration the set of genes with which each node is annotated (``node information''). Indeed, variable numbers of studies have been devoted to different biological concepts as well as the relationships among them. If a process has been studied in detail, it will be described by a larger number of nodes, at variable depth, the semantic difference between which can be much smaller than the difference between two nodes separated by only one edge in another portion of the GO DAG, related to a biological process on which only coarser information is available. This plays a role in choosing an appropriate strategy to reduce redundancy in the rejection set. 

To design the soft outer nodes filter, therefore, we down-weight a rejected node if it shares items with smaller rejected nodes. By contrast, if any of these smaller rejected nodes were descendants of this node, the outer nodes filter would have removed the node entirely. Given a rejection set $\cR$, each node receives a weight 
\begin{equation}
U_j = \gamma_j \cdot u_j,
\end{equation} 
where $u_j \geq 0$ is a weight determined a priori quantifying how much ``information" the node carries and $\gamma_j \in [0,1]$ quantifies the fraction of a node's information that is ``novel" with respect to other rejected nodes.

First, we give each node an a priori weight
\begin{equation}
u_j = -\log\left(\frac{|\mathcal A_j|}{K}\right),
\label{IC}
\end{equation}
where $\mathcal A_j$ are the items in node $j$ and $K$ is the total number of items. This quantity is the \textit{information content} (IC) of a node (see e.g. \cite{PetC09}), and is larger for smaller and thus more informative nodes. These IC weights establish a baseline for how important the discovery of a given node is. To determine $\gamma_j$, the idea is to evaluate distinct discoveries using the items they implicate: we give a discovered node $j \in \cR$ credit for the discovery of an item $k \in \mathcal A_j$ if it is the smallest among the discovered nodes containing $k$. Define
\begin{equation}
S_k = \min\{|\mathcal A_{j}|: j \in \cR, \ \mathcal A_{j} \ni k\},
\label{S_g_def}
\end{equation}
to be the minimum size of a discovered node containing item $k$, and define
\begin{equation*}
\cR^k = \{j \in \cR: \mathcal A_j \ni k, |\mathcal A_j| = S_k\},
\end{equation*}
the set of nodes achieving this minimum size. By convention, set $S_k = \infty$ and $\cR^k = \varnothing$ if no nodes containing item $k$ were discovered. Note that the set $\cR^k$ might have more than one node, indicating that there may be multiple nodes of the same size that all can claim credit for the discovery of item $k$. In this case, the credit is split equally among these nodes, each getting $1/|\cR^k|$. Putting these pieces together, the fraction of novel information contributed by node $j$ is
\begin{equation}
\gamma_j \equiv \frac{\sum_{k \in \mathcal A_j} \frac{1}{|\cR^k|}\ind(j \in \cR^k)}{|\mathcal A_j|};
\label{credit}
\end{equation}
i.e. the total credit the node received from all its items, divided by the total number of items in the node (see Figure \ref{fig:soft_outer_nodes}). Finally, putting together (\ref{IC}) and (\ref{credit}), the prioritization scores induced by the soft outer nodes filter are:
\begin{equation*}
U_j \equiv -\log\left(\frac{|\mathcal A_j|}{K}\right) \cdot \frac{\sum_{k \in \mathcal A_j} \frac{1}{|\cR^k|}\ind(j \in \cR^k)}{|\mathcal A_j|}.
\end{equation*}

\begin{figure}[h!]
	
	\begin{center}
		\hspace*{-.5cm}   \begin{tikzpicture}[scale=.85]
		
		\node[text=red] at(0.5,2) {0.75};
		\node[text=red] at(0.5,4) {0.33};
		\node[text=red] at(3.5,2) {0.75};
		\node[text=red] at(6.7,2) {0.5};
		\node[text=red] at(9.9,2) {1};
		\node[text=red] at(5.2,4) {0};
		\node[text=red] at(9.7,4) {1};
		\node[text=red] at(11.25,4) {1};
		\node[text=red] at(6.1,6) {$\displaystyle\frac{1}{8}$ \hspace{1in}};
		
		% \draw [help lines] (0,0) grid (13,8);
		\node[obs,  text opacity=1]  (H1)  at(0,0)   {a} ;
		\node[obs,  , text opacity=1]  (H2)  at(3,0)   {b} ;
		\node[obs,   text opacity=1]  (H3)  at(9,0)   {g} ;
		\node[obs,  text opacity=1,very thick,draw=red]  (-1H1)  at(1.5,2)   {{\bf a,b}$^*$} ;
		\node[obs, text opacity=1,very thick,draw=red]  (-1H2)  at(4.5,2)   {{\bf b}$^*$,{\bf e}} ;
		\node[obs,   text opacity=1,very thick,draw=red]  (-1H3)  at(7.5,2)   {{\bf f},g} ;
		\node[obs,  text opacity=1,very thick,draw=red]  (-1H4)  at(10.5,2)   {\bf g} ;
		\node[obs,  text opacity=1,very thick,draw=red]  (-2H1)  at(1.5,4)   {a,b,{\bf c}} ;
		\node[obs,  text opacity=1,very thick,draw=red]  (-2H2)  at(6,4)   {b,e,f,g} ;
		\node[obs,  text opacity=1]  (-2H3)  at(7.5,4)   { i} ;
		\node[obs,  text opacity=1]  (-2H4)  at(9,4)   {l} ;
		\node[obs,  text opacity=1,very thick,draw=red]  (-2H5)  at(10.5,4)   {\bf m,n,o} ;
		\node[obs,  text opacity=1,text width=.8cm,very thick,draw=red]  (-2H6)  at(12,4)   {\bf p,q,\\
			r,s,t} ;
		\node[obs,  text opacity=1]  (-3H1)  at(1.5,6)   {a,b,c,d} ;
		\node[obs,  text opacity=1,text width=1.2cm,very thick,draw=red]  (-3H2)  at(7.5,6)   {  a,b,c,\\ e,f,g,{\bf i,l}\\ m,n,o,p,\\ q,r,s,t} ;

		\edge[draw=gray!50] {-1H1} {H1}; %
		\edge[draw=gray!50]  {-1H1} {H2}; %

		\edge[draw=gray!50]  {-1H2} {H2}; %
		\edge[draw=gray!50]  {-1H3} {H3}; %
		\edge[draw=gray!50]  {-1H4} {H3}; %
		\edge[draw=gray!50]  {-2H1} {-1H1}; %
		\edge[draw=gray!50]  {-2H2} {-1H2}; %
		\edge[draw=gray!50]  {-2H2} {-1H3}; %
		\edge[draw=gray!50]  {-3H1} {-2H1}; %
		\edge[draw=gray!50]  {-3H2} {-2H1}; %
		\edge[draw=gray!50]  {-3H2} {-2H2}; %
		\edge[draw=gray!50]  {-3H2} {-2H3}; %
		
		\edge[draw=gray!50]  {-3H2} {-2H4}; %
		\edge[draw=gray!50]  {-3H2} {-2H5}; %
		\edge[draw=gray!50]  {-3H2} {-2H6}; %

		%  \plate[inner sep=0.15cm, xshift=-0.08cm, yshift=0.12cm] {plateX} {(H1) (H2) (H3) (H4) (H5) (H6) (H7) (H8) (H9) (H10) (H11) (H12) (H13) (H14)} {\textcolor{black}{Level 3}};
		
		\end{tikzpicture}
	\end{center}
	
	\caption{Example of soft outer node prioritization scores on a DAG. Nodes are circled in red if their corresponding hypotheses are in $\cR^*$.  Letters in the nodes represent the genes with which they are annotated; %for every node,
		genes for which the node gets credit are in bold, with an asterisk indicating if the credit is shared. The soft outer node prioritization score excluding the a priori IC weights (i.e. $\gamma_j$) is reported as a red number on the side of the node.}	
	\label{fig:soft_outer_nodes}
\end{figure}
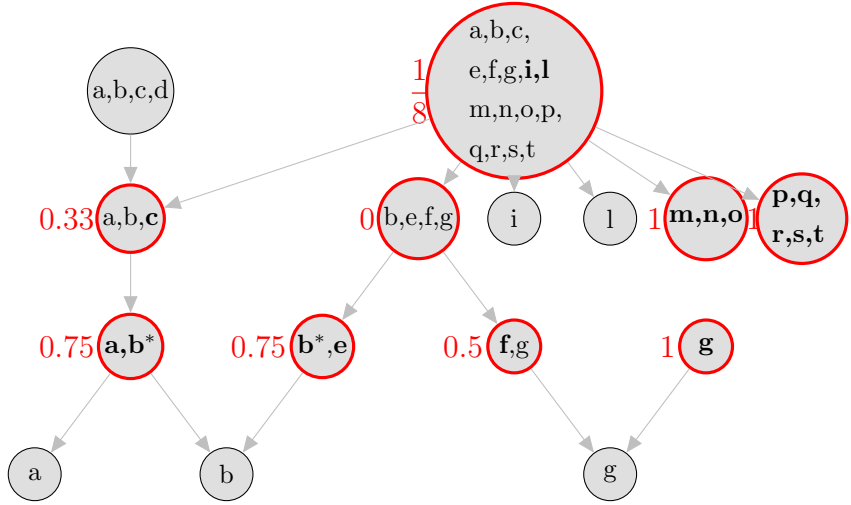

It turns out the soft outer nodes filter is monotonic on arbitrary DAGs. Indeed, to verify this, note that 
\begin{equation*}
\begin{split}
\norm{\F_0(\cR)} = \sum_{j = 1}^m U_j &= \sum_{j = 1}^m  \frac{\sum_{k \in \mathcal A_j} \frac{1}{|\cR^k|}\ind(j \in \cR^k)}{|\mathcal A_j|}u_j \\
&= \sum_{k = 1}^K \sum_{j \in \cR^k}\frac{u_j}{|\cR^k||\mathcal A_j|} = \sum_{k = 1}^K \sum_{j \in \cR^k}\frac{-\log\left(\frac{S_k}{K}\right)}{|\cR^k|S_k} = \sum_{k = 1}^K \frac{-\log\left(\frac{S_k}{K}\right)}{S_k}.
\end{split}
\end{equation*}
Clearly, from (\ref{S_g_def}), $S_g$ can only decrease as $\cR$ increases, so $\norm{\F_0(\cR)}$ increases as $\cR$ increases. Intuitively, the reason for the monotonicity is that expanding the rejection set $\cR$ can only decrease the size of the node(s) taking credit for a given item $k$, so the total weight accounted for by this item will increase.

Hence, if the node p-values are PRDS, the extension of Theorem \ref{main_theorem} part (i) to prioritization filters in Theorem~\ref{fractional_theorem} implies that Focused BH controls FDR on arbitrary DAGs with the soft outer nodes filter. Soft outer nodes are a meaningful summary for DAGs like GO where, in addition to edge information, it is important to consider node measures, in the form of associated genes. Finally, we note that the prioritization defined by soft outer nodes is a first example of $U_j \in [0,1]$, where the fractional value is not dictated by weights $u_j$ set a priori, but is data dependent.

\section{Connection to structured multiple testing} \label{sec:supp_structured}

As discussed in Section~\ref{sec:implications_beyond} of the main text, there is a connection between filtering and structured multiple testing. In this section, we consider fixed filters, which by slight abuse of notation we will think of as inputting only the set $\cR$ (since there is no dependency on $\bm p$):
\begin{equation}
\cU = \F(\cR).
\label{abuse}
\end{equation}

Consider the case of the outer nodes filter $\F = \F_{\mathcal T}$ with respect to a tree structure on $\mathcal H$. The reason we are applying a filter is to obtain a set $\cU^*$ that satisfies a particular ``non-redundancy" constraint:
\begin{equation}
\cU^* \in \sU \equiv \{\cU \subseteq \mathcal H: \text{there do not exist } j_1, j_2 \in \cU \text{ such that } j_1 \leadsto j_2\}.
\label{outer_nodes_structure}
\end{equation}
Here, $\mathscr U$ is the collection of sets satisfying the outer nodes property: no element is the descendant of another. When phrased this way, our problem becomes one of multiple testing with structural constraints:
searching for $\cR^*$ such that $\F_{\mathcal T}(\cR^*)$ controls FDR can be phrased as searching for $\cU^* \in \sU$ controlling FDR. We explore briefly this duality between structural constraints and filters.

For any filter we may define the class of acceptable rejection sets
\begin{equation}
\sU_{\F} \equiv \{\cU = \F(\cR) \text{ for some } \cR \subseteq \cH\}.
\label{U_F}
\end{equation}
The filter $\F$ then acts as a ``projection" onto this set $\sU_{\F}$. Then Focused BH is  a procedure that inputs a set of p-values $\bm p$ and outputs a member $\cU^* \in \sU_{\F}$ such that $\EE{\fdp(\cU^*)} \leq q$. 

Conversely, let $\sU \subseteq 2^{\cH}$ be a collection of rejection sets (specified a priori) that obey a certain structural constraint and let us consider a procedure that directly searches for a rejection set  $\cU^*\in \sU$.   Specifically, let us define  \textit{Structured BH} (independently discovered by \cite{RetJ17}) as follows. 
For each $\cU \in \sU$, define
\begin{equation}
\fdphat^{\text{SBH}}(\cU) \equiv \frac{m \cdot \max_{j \in \cU} p_j}{|\cU|}.
\label{fdphat}
\end{equation}\
Then, let
\begin{equation}
\cU^* \equiv \arg \max\{|\cU|: \cU \in \sU,\ \fdphat^{\text{SBH}}(\cU) \leq q\},
\label{structured_BH}
\end{equation}
i.e. we choose the largest set in $\sU$ for which the estimated FDP is below the target. Note that this maximum might not be unique, in which case we allow Structured BH to output any of these maximal sets. 

\begin{proposition} \label{prop:structured_BH} For any structure class $\sU$, Structured BH controls the FDR as long as the p-values are PRDS. 
\end{proposition}
\begin{proof}
	This statement follows from the first part of Lemma~\ref{lem:FDR_control} since Structured BH is by construction a self-consistent procedure with respect to $\Delta(\bm p, r) = \frac{q}{m}r$, for which $|\mathcal U^*|$ is a nonincreasing function of each p-value.
\end{proof}

What is the relation between these two approaches to the problem? 
It turns out that Structured BH for any $\sU$ can be recast as Focused BH for the filter $\F_{\sU}(\cR)$, defined via
%\begin{equation*}
%\F_{\sU}(\cR) \equiv \text{ subsetting filter with rejection region } \underset{\cU \in \sU: \cU \subseteq \cR}{\arg \max}\ |\cU|,
%\end{equation*}
\begin{equation}
\F_{\sU}(\cR)\equiv  \underset{\cU \in \sU: \cU \subseteq \cR}{\arg \max}\ |\cU|,
\label{filter_structure}
\end{equation}
i.e. $\F_{\sU}(\cR)$ is the largest subset of $\cR$ belonging to $\sU$ (again, we allow any maximal element above if there is not a unique one).  

\begin{proposition} \label{prop:equivalence}
	We are given a structure class $\sU$. Let us define a filter $\F_{\sU}(\cR)$ via \eqref{filter_structure}. Then, Structured BH with the structure class $\sU$ is equivalent to Focused BH with the filter $\F_{\sU}$ (modulo the possible ambiguity in the definitions of Structured BH and $\F_{\sU}$).
\end{proposition}
\begin{proof}
	Let $\cU^*$ be the output of Structured BH, and let $t^*$ be the Focused BH threshold. For the proof we abbreviate $\F_{\sU}$ by writing $\F$ instead. We must show that $\fdphat^{\text{SBH}}(\F(\cR(t^*))) \leq q$ and $|\F(\cR(t^*))| = |\cU^*|$, which would imply that $\F(\cR(t^*))$ is one of the potential outputs of Structured BH. 
	
	Let $t^{**} = \max_{j \in \cU^*} p_j$. We claim that $|\cU^*| = |\F(\cR(t^{**}))|$. To see this, note that $\F(\cR(t^{**})) \in \sU$ and $|\F(\cR(t^{**}))| \geq |\cU^*|$ by the definition of $\F$. Additionally,
	\begin{equation}
	\fdphat^{\text{SBH}}(\F(\cR(t^{**}))) = \frac{m \cdot \max_{j \in \F(\cR(t^{**}))}p_j}{|\F(\cR(t^{**}))|} \leq \frac{m \cdot t^{**}}{|\cU^*|} = \fdphat^{\text{SBH}}(\cU^*) \leq q.
	\end{equation}
	Hence, $\F(\cR(t^{**}))$ is a set of maximal size in $\sU$ for which $\fdphat^{\text{SBH}} \leq q$, so $|\F(\cR(t^{**}))| = |\cU^*|$. Next, we claim that $|\F(\cR(t^{**}))| = |\F(\cR(t^{*}))|$. We have
	\begin{equation*}
	\fdphat^{\text{FBH}}(t^{**}) = \frac{m \cdot t^{**}}{|\F(\cR(t^{**}))|} = \frac{m \cdot t^{**}}{|\cU^*|} \leq q,
	\end{equation*}
	from which it follows that $t^{*} \geq t^{**}$. Since $\F$ is monotonic, it follows that $|\F(\cR(t^{*}))| \geq |\F(\cR(t^{**}))|$. On the other hand, note that 
	\begin{equation}
	\fdphat^{\text{SBH}}(\F(\cR(t^*))) = \frac{m\cdot \max_{j \in \F(\cR(t^*))}p_j}{|\F(\cR(t^*))|} \leq \frac{m\cdot t^*}{|\F(\cR(t^*))|} = \fdphat^{\text{FBH}}(t^*) \leq q,\label{t-star-cond}
	\end{equation}
	which shows that $|\F(\cR(t^{**}))| = |\cU^*|\geq |\F(\cR(t^{*}))|$. Hence, $|\F(\cR(t^{**}))| = |\F(\cR(t^{*}))|$, as claimed. It follows that $|\cU^*| = |\F(\cR(t^*))|$, which completes the proof.
\end{proof}

Conversely, if a fixed and monotonic filter is also \textit{idempotent} (i.e. $\F^2 = \F$), then $\F = \F_{\sU}$ for the structure class $\sU$ defined in (\ref{U_F}). 

\begin{proposition} \label{prop:equivalence_2}
	Suppose $\F$ is a fixed, monotonic, and idempotent filter. Then, $\F = \F_{\sU}$ for $\sU = \{\F(\cR): \cR \subseteq 2^{m}\}$.
\end{proposition}
\begin{proof}
	Fix $\cR \subseteq 2^m$. We must show that $\F(\cR)$ is a maximal subset of $\cR$ among all sets in $\sU$. In other words, we must show that $|\F(\cR)| \geq |\F(\cR')|$ for all $\cR'$ such that $\cR \supseteq \cU$ when $\cU$ is defined by $\F(\cR') $. This claim holds due to idempotency and monotonicity, since $\cU \subseteq \cR$ implies that $|\F(\cR')| = |\F(\F(\cR'))| \leq |\F(\cR)|$.
\end{proof}

Thus, there is a close relationship between filtering and enforcing structural constraints in multiple testing. 

It is then of interest to compare  Structured BH  to other procedures for structured multiple testing. Specifically, STAR \cite{LetF17} also controls FDR under arbitrary structural constraints. STAR proceeds by first constructing a hypothesis ordering $\pi(1), \pi(2), \dots, \pi(m)$ so that non-nulls are likely to occur near the beginning of the ordering and so that candidate rejection sets $\cU_j \equiv \{\pi(1), \dots, \pi(j)\} \in \sU$ for all $j \in [m]$. Once this ordering is constructed, the final rejection set $\cU^* = \cU_{j^*}$ is chosen via an accumulation test \cite{LB17}. 

STAR works well for structure classes $\sU$ that can be built up recursively, such as convex regions or subtrees (see \cite{LetF17} for these and other examples). However, the methodology is not designed for structure classes of the kind we consider here, such as sets of nodes in a DAG satisfying the outer nodes property. Indeed, it is not clear how one would choose an ordering of nodes so that the first $k$ always have the outer nodes property. On the other hand, STAR can boost power if the structure is informative, whereas Structured BH cannot. 
In summary, despite the superficial similarity between Structured BH and STAR, these two methods address different sets of problems.

\end{document}